\documentclass[
reprint,
amsmath,amssymb,
aps, pre,
]{revtex4-2}

\usepackage[title]{appendix}
\usepackage{graphicx}
\usepackage{dcolumn}
\usepackage{bm}
\usepackage{bbold}
\usepackage[hidelinks]{hyperref}
\usepackage[capitalise]{cleveref}
\crefname{appendix}{Appendix}{Appendices}
\Crefname{appendix}{Appendix}{Appendices}
\usepackage{lipsum}
\usepackage[dvipsnames]{xcolor}
\usepackage{amsfonts}
\usepackage{yfonts}
\usepackage{booktabs}

\usepackage{amsthm}
\usepackage[scr=boondoxo,scrscaled=1.05]{mathalfa}
\theoremstyle{definition}
\newtheorem{lemma}{Lemma}

\newtheorem{theorem}{Theorem}

\newcommand{\rmd}{\mathrm{d}}

\DeclareFontFamily{U}{futm}{}
\DeclareFontShape{U}{futm}{m}{n}{
  <-> s * [.97534] fourier-bb 
  }{}
\DeclareMathAlphabet{\mathbbs}{U}{futm}{m}{n}

\begin{document}

\preprint{APS/123-QED}

\title{Mutual linearity in and out of stationarity for Markov jump processes: A trajectory-based approach}
\author{Jiming Zheng} \email{jiming@unc.edu}
\author{Zhiyue Lu} \email{zhiyuelu@unc.edu}
\affiliation{Department of Chemistry, University of North Carolina-Chapel Hill, NC}

\date{\today}

\begin{abstract}
Nonequilibrium response theory is a fundamental framework for understanding how physical systems respond to perturbations. Recently, a mutual linearity has been discovered for Markov jump processes using linear algebra analysis. This mutual linearity states that two observables are linearly dependent on each other in the long-time limit when the transition rate of a single edge is altered. It has also been extended to non-stationary cases for current observables. In this work, we provide a trajectory-based derivation of mutual linearity utilizing the trajectory-level linear response theory. The trajectory approach allows us to generalize the mutual linearity to non-stationary relaxation dynamics for state observables and counting observables. Our results shed light on the fundamental response properties far from equilibrium and the trajectory-level origin of mutual linearity. Our trajectory-based approach makes it possible to generalize the mutual linearity to a broader class of systems, including diffusion processes and open quantum systems.
\end{abstract}

\maketitle

\section{Introduction}

Understanding nonequilibrium systems through their response to perturbations is a central problem in physics. Response theory provides a direct route to probe driven systems and has found wide applications in complex physical and biological systems, including sensitivity \cite{mora2015physical,bialek2005physical,hartich2016sensory}, adaptation \cite{wark2007sensory,lan2012energy,conti2022nonequilibrium}, and robustness \cite{pittendrigh1954temperature,johnson2021circadian,hogenesch2011understanding,ay2007geometric,fu2024temperature}. A particularly powerful approach to nonequilibrium response is the trajectory-level description of stochastic dynamics. Early work by Christian Maes and collaborators established that the linear response can be formulated directly at the level of stochastic trajectories \cite{maes2020response,maes2013fluctuation,baiesi2009fluctuations,baiesi2013update}, revealing a general structure of response beyond equilibrium. From the trajectory perspective, a variety of universal response relations \cite{seifert2010fluctuation,pagare2024stochastic,zheng2025nonequilibrium} and fluctuation-response inequalities \cite{dechant2020fluctuation,zheng2025nonlinear,zheng2025universal,zheng2025unified,zheng2026thermodynamic,kwon2025fluctuation,lee2021universal,dechant2025finite,hasegawa2019uncertainty,van2025fundamental} have been uncovered in recent years. In parallel, linear algebraic methods \cite{aslyamov2025nonequilibrium,aslyamov2026macroscopic,ptaszynski2026nonequilibrium,harunari2024mutual,bebon2026mutual} and matrix tree theorems \cite{harunari2024mutual,bebon2026mutual,chun2023trade,fernandes2023topologically,owen2020universal} have been used to derive exact response relations and inequalities for Markov jump processes, providing a complementary perspective on nonequilibrium responses.

Among recent developments, a remarkable property known as \emph{mutual linearity} has been identified for Markov jump processes \cite{harunari2024mutual,floyd2025local,bebon2026mutual}. It states that when a single transition rate is perturbed, different steady-state observables become linearly dependent. More recently, this structure has been extended to multi-edge perturbations \cite{dal2025mutual}, and interpreted through first-passage times \cite{khodabandehlou2025affine}. Although this result has been derived using linear algebraic methods at the generator level \cite{aslyamov2024general} or the network analysis \cite{polettini2025coplanarity}, its physical origin from the trajectory perspective remains unclear. In particular, it is not evident why responses of distinct trajectory observables should exhibit such a universal linear structure.

In this work, we develop a trajectory-level derivation of mutual linearity based on the Doob-Meyer decomposition. Within this framework, linear response is expressed in terms of correlations with martingale noise, directly linking response properties to stochastic fluctuations along trajectories. We show that mutual linearity arises from a simple multiplicative structure in the response kernel associated with transition probabilities. This perspective clarifies the trajectory-level origin of mutual linearity and provides a transparent interpretation of its universality.

Moreover, the trajectory-based formulation allows us to extend the mutual linearity relation from steady states into non-stationary relaxation dynamics governed by time-independent dynamical rates. By analyzing the Laplace-transformed response, we demonstrate that the linear dependence between observables persists in the frequency domain for non-stationary relaxation dynamics. This reveals that mutual linearity is a dynamical property of the response rather than a feature restricted to steady-state averages. Finally, since trajectory-level techniques are well developed for diffusion processes and open quantum systems, our approach provides a possible route to extending mutual linearity to continuous and quantum stochastic systems.

The paper is organized as follows. In \cref{sec: doob-meyer}, we introduce the Doob-Meyer decomposition and derive correlation functions of trajectory increments. In \cref{sec: response correlations}, we formulate linear response in terms of these correlations. We then establish mutual linearity in the long-time limit in \cref{sec: mutual linearity stationary}. Next, we extend the analysis to the frequency domain and demonstrate non-stationary mutual linearity in \cref{sec: mutual linearity frequency}. Finally, we present numerical verification using simple exclusion process models in \cref{sec: numerical}.

\section{Doob-Meyer decomposition for Markov jump processes} \label{sec: doob-meyer}

In this work, we focus on Markov jump processes, which are a class of stochastic processes that evolve in continuous time and have discrete state spaces. A Markov jump process on $N$ states can be described by the master equation:
\begin{equation}
  \frac{\rmd \boldsymbol{p}(t)}{\rmd t} = R \boldsymbol{p}(t),
\end{equation}
where $\boldsymbol{p}(t) = (p_1(t), p_2(t), \cdots, p_N(t))^\top$ is the probability distribution over the states and $R = \{r_{ij}\}_{N\times N}$ is the transition rate matrix. The off-diagonal elements $r_{ij}$ represent the transition rates from state $j$ to state $i$, while the diagonal elements are defined as $r_{ii} = -\sum_{j \neq i} r_{ji}$ to ensure that the total probability is conserved. We assume the system is irreducible; thus, its steady state is unique. The steady state of the system is given by the eigenvector $\boldsymbol{\pi}$ of $R$ corresponding to the eigenvalue $0$, which satisfies $R \boldsymbol{\pi} = 0$. Throughout this work, we restrict our discussion to Markov jump processes with a finite number of states $N$.

A stochastic trajectory of the Markov jump process can be represented as a sequence of states and transition times, denoted by $X_\tau = \{(x_k, t_k)\}_{k=0}^n$, where $x_k$ is the state at time $t_k$ and $n$ is the total number of transitions up to time $\tau$. For a stochastic trajectory $X_\tau$, we can define a counting process $n_{ij}(\tau)$ that counts the number of transitions from state $j$ to state $i$ up to time $\tau$. It can be written as a time integral of $\rmd n_{ij}(t)$,
\begin{equation}
  n_{ij}(\tau) = \int_0^\tau \rmd n_{ij}(t),
\end{equation}
where $\rmd n_{ij}(t)$ is an increment equal to 1 if there is a transition from state $j$ to state $i$ at time $t$, and 0 otherwise. The counting process $n_{ij}(\tau)$ is a stochastic process that captures the number of transitions between states over time.

The transition number $n_{ij}(\tau)$ is a submartingale, i.e.,
\begin{equation}
  \langle n_{ij}(\tau) \mid X_{\tau'} \rangle \geq n_{ij}(\tau'), \quad \text{for } \tau \geq \tau',
\end{equation}
where $\langle A \mid B \rangle$ denotes the expectation of $A$ conditioned on the specified value of $B$. The Doob-Meyer decomposition \cite{meyer1962decomposition} allows us to express the submartingale $n_{ij}(\tau)$ as the sum of a predictable compensator and a martingale:
\begin{equation}
  n_{ij}(\tau) = \int_0^\tau r_{ij} \rmd \tau_j(t) + \varepsilon_{ij}(\tau),
\end{equation}
where $\int_0^\tau r_{ij} \rmd \tau_j(t)$ is left-continuous and thus mathematically predictable, $\rmd\tau_j(t)$ is the time spent in state $j$ during the interval $[0, \tau]$, and $r_{ij} \rmd \tau_j(t)$ represents the expected number of transitions from state $j$ to state $i$ during this time. The term $\varepsilon_{ij}(\tau)$ is a martingale that captures the fluctuations around this expected value. Its derivative is defined as the centered Poisson noise $\rmd \varepsilon_{ij}(t) \equiv \rmd n_{ij}(t) - r_{ij} \rmd \tau_j(t)$, which has zero mean and captures the stochastic fluctuations in the counting process.

The Doob-Meyer decomposition offers a stochastic differential equation representation of the Markov jump process:
\begin{equation}
  \rmd n_{ij}(t) = r_{ij} \rmd \tau_j(t) + \rmd \varepsilon_{ij}(t).
  \label{eq: SDE master equation}
\end{equation}
This equation has been used recently to derive thermodynamic uncertainty relations for Markov jump processes \cite{stutzer2025stochastic}. Here, we reformulate it in the language of Doob-Meyer decomposition. \cref{eq: SDE master equation} takes a similar form as the Langevin equation for diffusion processes. However, unlike Gaussian noise in the Langevin equation, the noise term $\rmd \varepsilon_{ij}(t)$ in \cref{eq: SDE master equation} obeys a Poisson distribution. $\rmd \varepsilon_{ij}(t)$ has zero mean and its correlation function is given by $\langle \rmd \varepsilon_{ij}(t) \rmd \varepsilon_{kl}(t') \rangle = \delta_{ik} \delta_{jl} \delta(t-t') r_{ij} p_j(t) \rmd t \rmd t'$, which follows from the property that the variance of a Poisson process is equal to its mean.

Now we examine correlations among the noise term $\rmd\varepsilon_{ij}(t)$, the dwelling time $\rmd\tau_j(t)$, and the transition number $\rmd n_{ij}(\tau)$. These can be straightforwardly computed using the probability distribution of $\rmd n_{ij}(t)$, $\rmd \tau_j(t)$, and the Markov property. One may refer \cite{stutzer2025stochastic} for a detailed proof.

\begin{lemma}(Noise-Noise correlations) \label{lemma: noise-noise}
  The correlations between $\rmd\varepsilon_{ij}(t)$ and $\rmd\varepsilon_{kl}(t')$ is given by
  \begin{equation}
      \langle \rmd\varepsilon_{ij}(t) \rmd\varepsilon_{kl}(t') \rangle = \delta_{ik} \delta_{jl} \delta(t-t') r_{ij} p_j(t) \rmd t \rmd t'.
  \end{equation}
\end{lemma}

This lemma states that the noise terms $\rmd\varepsilon_{ij}(t)$ and $\rmd\varepsilon_{kl}(t')$ are uncorrelated for $t \neq t'$, and their correlation at the same time is given by $\delta_{ik} \delta_{jl} \delta(t-t') r_{ij} p_j(t) \rmd t$. This is a direct consequence of the properties of the Poisson process.

\begin{lemma}(Noise-Time correlations) \label{lemma: noise-time}
  The correlations between $\rmd\varepsilon_{ij}(t)$ and $\rmd\tau_k(t')$ is given by
  \begin{align}
      \langle \rmd\varepsilon_{ij}(t) \rmd\tau_k(t') \rangle = \mathbb{1}_{t<t'} [P(k, t'| i, t)-P(k, t'| j, t)] \nonumber \\
      \times r_{ij}p_j(t)\rmd t\rmd t'.
  \end{align}
\end{lemma}

\cref{lemma: noise-time} captures the correlation between the noise term $\rmd\varepsilon_{ij}(t)$ and the dwelling time $\rmd\tau_k(t')$ at a later time $t'$. The indicator function $\mathbb{1}_{t<t'}$ ensures that there is no correlation when $t \geq t'$, which is consistent with the martingale property $\langle \rmd \varepsilon_{ij}(t) \mid X_{t'} \rangle = 0$. The term $P(k, t'| i, t)$ represents the probability of being in state $k$ at time $t'$ given that the system was in state $i$ at time $t$, and similarly for $P(k, t'| j, t)$. The difference between these two probabilities captures how the noise term $\rmd\varepsilon_{ij}(t)$ influences the future state of the system, which in turn affects the dwelling time $\rmd\tau_k(t')$.

\begin{lemma}(Noise-Jump correlations) \label{lemma: noise-jump}
  The correlation between $\rmd\varepsilon_{ij}(t)$ and $\rmd n_{kl}(t')$ is given by
  \begin{align}
      \langle \rmd\varepsilon_{ij}(t) \rmd n_{kl}(t') \rangle = \mathbb{1}_{t<t'} r_{kl} [P(l, t'| i, t)-P(l, t'| j, t)] \nonumber \\
      \times r_{ij}p_j(t)\rmd t\rmd t' + r_{ij}p_j(t)\delta_{ik}\delta_{jl}\delta(t-t')\rmd t.
  \end{align}
\end{lemma}

\cref{lemma: noise-jump} is obtained by combining \cref{lemma: noise-noise} and \cref{lemma: noise-time} with the definition $\rmd \varepsilon_{ij} = \rmd n_{ij} - r_{ij}\rmd t_j$. The first term on the right-hand side of \cref{lemma: noise-jump} captures the correlation between $\rmd\varepsilon_{ij}(t)$ and $\rmd n_{kl}(t')$ when $t<t'$, and the second term captures the correlation when $t=t'$. When $t>t'$, there is no correlation between $\rmd\varepsilon_{ij}(t)$ and $\rmd n_{kl}(t')$ due to the martingale property.

\section{Trajectory level linear response theory} \label{sec: response correlations}

To provide a proof for the mutual linearity directly based on the trajectory-level response theory, we first revisit the linear response theory for Markov jump processes derived in \cite{seifert2010fluctuation,zheng2025unified}. Here, the term ``linear response'' refers to the first derivative of an observable with respect to a control parameter around a given reference dynamics. The reference dynamics is not assumed to be close to equilibrium; it can be an arbitrary nonequilibrium steady state or a transient state. Thus, the response theory used below is linear only in the perturbation amplitude, not in the sense of near equilibrium systems.

The stochastic trajectory $X_\tau$ of a Markov jump process can be described by the sequence of states and transition times $X_\tau = \{(x_k, t_k)\}_{k=0}^n$. The probability of observing a particular trajectory $X_\tau$ can be expressed in terms of the transition rates and the dwelling times in each state. Specifically, the probability density of a trajectory can be written as \cite{peliti2021stochastic}:
\begin{align}
  \mathcal{P}[X_\tau] &= p_{x_0}(0) \prod_{k=1}^n r_{x_k x_{k-1}} \prod_{k=0}^n e^{\int_{t_k}^{t_{k+1}} r_{x_k x_k} \rmd t} \\
  &= p_{x_0}(0) \exp\sum_{i \neq j} \left[ \int_0^\tau \rmd n_{ij}(t) \ln r_{ij} - r_{ij} \rmd \tau_j(t) \right],
\end{align}
where $p_{x_0}(0)$ is the initial probability of being in state $x_0$, $t_{n+1} = \tau$ is the final time, and we group the terms involving the transition of dwelling events on the same state together in the second line. A trajectory observable $Q[X_\tau]$ is a functional of the trajectory, which can depend on the sequence of states and transition times. The expectation value of $Q$ can be computed as:
\begin{equation}
  \langle Q(\tau) \rangle = \int \mathcal{D}[X_\tau] \mathcal{P}[X_\tau] Q[X_\tau],
\end{equation}
where $\int \mathcal{D}[X_\tau]$ denotes the integration over all possible trajectories.

The linear response is defined as the first-order derivative of $\langle Q(\tau) \rangle$ with respect to a perturbation parameter $\lambda$ that modifies the transition rates $r_{ij}(\lambda)$. For example, the parameter $\lambda$ can influence the temperature of the heat bath, the chemical potential of the particle reservoir, or the external force applied to the system. We assume that the initial distribution is fixed and independent of the perturbation parameter. Additionally, the perturbed and unperturbed systems share the same initial distribution. The linear response can be computed using the path integral representation of $\langle Q(\tau) \rangle$:
\begin{equation}
  \frac{\rmd \langle Q(\tau) \rangle}{\rmd \lambda} = \int \mathcal{D}[X_\tau] \mathcal{P}[X_\tau] \left( \frac{\rmd \ln \mathcal{P}[X_\tau]}{\rmd \lambda} Q[X_\tau] \right),
\end{equation}
where we assume that the observable $Q$ does not explicitly depend on $\lambda$. The term $\frac{\rmd \ln \mathcal{P}[X_\tau]}{\rmd \lambda}$ can be computed from the expression of $\mathcal{P}[X_\tau]$:
\begin{align}
  \frac{\rmd \ln \mathcal{P}[X_\tau]}{\rmd \lambda} &= \sum_{i \neq j} \left[ \int_0^\tau \rmd n_{ij}(t) \frac{\rmd \ln r_{ij}}{\rmd \lambda} - \int_0^\tau \frac{\rmd r_{ij}}{\rmd \lambda} \rmd \tau_j(t) \right] \nonumber \\
  &= \int_0^\tau \sum_{i \neq j} \frac{\rmd \ln r_{ij}}{\rmd \lambda} \rmd\varepsilon_{ij}(t).
\end{align}
Therefore, the linear response can be expressed in terms of the correlation between the observable $Q$ and the noise term $\rmd\varepsilon_{ij}(t)$:
\begin{equation}
  \frac{\rmd \langle Q(\tau) \rangle}{\rmd \lambda} = \sum_{i \neq j} \int_0^\tau \frac{\rmd \ln r_{ij}}{\rmd \lambda} \langle Q(\tau) \rmd \varepsilon_{ij}(t) \rangle.
\end{equation}
This expression has been obtained in \cite{zheng2025unified}. It shows that the linear response of the observable $Q$ to the perturbation can be computed as the correlation between $Q$ and the noise term $\rmd\varepsilon_{ij}(t)$, which captures the fluctuations in the counting process around its expected value. Specifically, for $\lambda = r_{ij}$, we have $\frac{\rmd \ln r_{ij}}{\rmd \lambda} = \frac{1}{r_{ij}}$, and the linear response becomes:
\begin{equation}
  \frac{\rmd \langle Q(\tau) \rangle}{\rmd r_{ij}} = \int_0^\tau \frac{1}{r_{ij}} \langle Q(\tau) \rmd \varepsilon_{ij}(t) \rangle.
\end{equation}

\section{Mutual linearity in steady states} \label{sec: mutual linearity stationary}

In this work, we consider trajectory observables that are linear functionals of dwelling times and counting statistics, such as:
\begin{equation}
  Q[X_\tau] = \int_0^\tau \sum_k a_k \rmd \tau_k(t) + \sum_{k \neq l} b_{kl} \rmd n_{kl}(t),
\end{equation}
where $a_{k}$ and $b_{kl}$ are coefficients that determine how the observable depends on the transitions and dwelling times. The observable becomes a state observable when $b_{kl} = 0$ for all $k, l$, and becomes a counting observable when $a_k = 0$ for all $k$. The current observable is a special case of the counting observable where $b_{kl} = -b_{lk}$ for all $k,l$. The steady state mutual linearity states that two observables $Q_1$ and $Q_2$ are linearly dependent on each other in the long time limit, which means that the amount of change in $\langle Q_1(\tau) \rangle$ due to a perturbation on the transition rate $r_{ij}$ is proportional to the amount of change in $\langle Q_2(\tau) \rangle$ due to the same perturbation, with a proportionality constant that is independent of the perturbation on $r_{ij}$. This result has been derived in \cite{harunari2024mutual,bebon2026mutual} using linear algebraic methods. Here, we provide an alternative proof based on the trajectory-level linear response theory.

For simplicity, we first consider two state observables $Q_1[X_\tau] = \int_0^\tau \rmd \tau_k(t)$, $Q_2[X_\tau] = \int_0^\tau \rmd \tau_{l}(t)$, and the perturbation on transition rate $r_{ij}$ where $i$ and $j$ are different from $k$ and $l$. Using \cref{lemma: noise-time}, the linear response of $\langle Q_1(\tau) \rangle$ is given by:
\begin{subequations}
\begin{align}
  &\frac{\rmd \langle Q_1(\tau) \rangle}{\rmd r_{ij}} = \int_0^\tau \rmd t \int_0^\tau \rmd t' \frac{1}{r_{ij}} \frac{\langle \rmd \tau_k(t') \rmd \varepsilon_{ij}(t) \rangle}{\rmd t' \rmd t} \\
  ={}& \int_0^\tau \rmd t \int_t^\tau \rmd t' ~ p_j(t) [P(k, t'| i, t) - P(k, t'| j, t)],
\end{align}
\end{subequations}

We can change the variable from $t'$ to $s = t'-t$ and rewrite the linear response $\rmd_{r_{ij}}\langle Q_1(\tau) \rangle$ as:
\begin{equation}
   \int_0^\tau \rmd t \int_0^{\tau-t} \rmd s ~ p_j(t) [P(k, s+t| i, t) - P(k, s+t| j, t)].
\end{equation}
For the system with time-independent rates, the transition probabilities only depend on the time difference $t'-t$. Therefore, we have $P(k, s+t| i, t) \approx P(k, s| i)$ and $P(k, s+t| j, t) \approx P(k, s| j)$. Furthermore, in the long time limit $\tau \to \infty$, the state probability $p_j(t)$ becomes time-independent, $\pi_j$. Thus, the linear response in the long-time limit can be computed as:
\begin{align}
  &\lim_{\tau \to \infty} \frac{\rmd \langle Q_1(\tau) \rangle}{\rmd r_{ij}} \nonumber \\
  ={}& \lim_{\tau \to \infty} \int_0^{\tau-s} \rmd t \int_0^\tau \rmd s ~ p_j(t) [P(k, s| i) - P(k, s| j)] \\
  ={}& \lim_{\tau \to \infty} \pi_j \int_0^\tau \rmd s ~ (\tau-s) [P(k, s| i) - P(k, s| j)] \\
  ={}& \lim_{\tau \to \infty} \pi_j \tau \int_0^\tau \rmd s ~ [P(k, s| i) - P(k, s| j)] + \mathcal{O}(1),
\end{align}
where we have used the fact that $p_j(t) \to \pi_j$ as $t \to \infty$ in the steady state. The term $(\tau-s)$ arises from the integral $\int_0^{\tau-s} \rmd t$. In the last line, we have separated the leading order term that scales with $\tau$ and the subleading term of order $\mathcal{O}(1)$. The linear response of $\langle Q_2(\tau) \rangle$ can be obtained by replacing $k$ with $l$ in the above expression. 

The fundamental matrix $Z$ is defined as
\begin{equation}
  Z = \int_0^\infty \left( e^{Rt} - \boldsymbol{\pi}\boldsymbol{1}^\top \right) \rmd t,
\end{equation}
where $\boldsymbol{\pi}$ is the column vector of the steady state distribution and $\boldsymbol{1}^\top$ is the row vector of ones. With this definition, we can express the ratio of the response in the long-time limit as
\begin{subequations}
\begin{align}
  \lim_{\tau \to \infty}\frac{\rmd_{r_{ij}} \langle Q_1(\tau) \rangle}{\rmd_{r_{ij}} \langle Q_2(\tau) \rangle} &= \frac{\int_0^\infty [P(k, s| i) - P(k, s| j)] \rmd s}{\int_0^\infty [P(l, s| i) - P(l, s| j)] \rmd s} \\
  &= \frac{Z_{ki} - Z_{kj}}{Z_{li} - Z_{lj}} \\
  &\equiv \chi_{ij}^{kl}.
\end{align}
\end{subequations}

We now state the key result that the ratio $\chi_{ij}^{kl}$ is independent of the perturbed transition rate $r_{ij}$. The detailed proof is in \cref{AppSec:proof_of_theorem1}.
\begin{theorem} \label{theorem: dchi_ij^kl/dr=0}
  For observables that excludes the transition from $j$ to $i$, the ratio $\chi_{ij}^{kl}$ is independent of the transition rate $r_{ij}$, i.e., $\rmd \chi_{ij}^{kl} / \rmd r_{ij} = 0$.
\end{theorem}

\cref{theorem: dchi_ij^kl/dr=0} shows that the ratio $\chi_{ij}^{kl}$ is independent of the transition rate $r_{ij}$, which implies that the linear responses of $\langle Q_1(\tau) \rangle$ and $\langle Q_2(\tau) \rangle$ to the perturbation on $r_{ij}$ are proportional in the long-time limit. Therefore, on any connected interval of the parameter $r_{ij}$ where $\rmd_{r_{ij}} \langle Q_2(\infty) \rangle$ does not vanish, the ratio of the linear responses is constant and equal to $\chi_{ij}^{kl}$. This gives the mutual linearity between $Q_1$ and $Q_2$ in the steady state:
\begin{equation}
  \left.\langle Q_1(\infty) \rangle\right|_{r_{ij}} = \chi_{ij}^{kl} \left.\langle Q_2(\infty) \rangle\right|_{r_{ij}} + \gamma_{ij}^{kl},
  \label{eq: mutual linearity dwelling}
\end{equation}
where $\langle Q_1(\infty) \rangle \equiv \lim_{\tau \to \infty} \frac{1}{\tau} \langle Q_1(\tau) \rangle$, $\langle Q_2(\infty) \rangle \equiv \lim_{\tau \to \infty} \frac{1}{\tau} \langle Q_2(\tau) \rangle$, and $\gamma_{ij}^{kl}$ is the intercept. The mutual linearity implies that the two observables $Q_1$ and $Q_2$ contain the same information about the system's response to the perturbation on $r_{ij}$. This result \cref{eq: mutual linearity dwelling} recovers the mutual linearity of the steady state probability distribution in \cite{bebon2026mutual} from the stochastic trajectory point of view.

Now we consider general time-averaged state-counting observables:
\begin{equation}
  Q_m[X_\tau] =  \int_0^\tau \sum_l a_l^{(m)} \rmd \tau_l(t) + \sum_{k \neq l / j \to i} b_{kl}^{(m)} \rmd n_{kl}(t),
  \label{eq: state-counting observable}
\end{equation}
where the summation in the second term excludes the transition from state $j$ to state $i$. The linear response of $\langle Q_m(\tau) \rangle$ to the perturbation on $r_{ij}$ can be computed by combining \cref{lemma: noise-time} and \cref{lemma: noise-jump}. In the long time limit, the linear response of $\langle Q_m(\tau) \rangle$ is given by
\begin{widetext}
\begin{subequations}
\begin{align}
  \lim_{\tau \to \infty} \frac{\rmd \langle Q_m(\tau) \rangle}{\rmd r_{ij}} ={}& \lim_{\tau \to \infty} \int_0^\tau \rmd t \int_0^\tau \rmd t' \left[ \sum_l a_l^{(m)} \frac{\langle \rmd \tau_l(t') \rmd \varepsilon_{ij}(t) \rangle}{\rmd t' \rmd t} + \sum_{k \neq l / j \to i} b_{kl}^{(m)} \frac{\langle \rmd n_{kl}(t') \rmd \varepsilon_{ij}(t) \rangle}{\rmd t' \rmd t} \right] \\
  ={}& \lim_{\tau \to \infty} \int_0^\tau \rmd t \int_t^\tau \rmd t' \sum_l \left[ a_l^{(m)} + \sum_{k (\neq l) / j \to i} b_{kl}^{(m)} r_{kl} \right] p_j(t) [P(l, t'| i, t) - P(l, t'| j, t)]
  \label{eq: linear response} \\
  ={}& \lim_{\tau \to \infty} \pi_j \tau \sum_l c_l^{(m)} (Z_{li} - Z_{lj}) + \mathcal{O}(1),
\end{align}
\end{subequations}
\end{widetext}
where we defined $c_l^{(m)} = a_l^{(m)} + \sum_{k (\neq l) / j \to i} b_{kl}^{(m)} r_{kl}$ for simplicity, where the sum $\sum_{k (\neq l) / j \to i}$ is over all transitions from state $l$ to state $k$ except the transition from state $j$ to state $i$. It is worth noticing that $c_l$ is independent of the transition rate $r_{ij}$ but is dependent on all the other transition rates. The ratio of the linear response of $\langle Q_1(\tau) \rangle$ and $\langle Q_2(\tau) \rangle$ to the perturbation on $r_{ij}$ in the long time limit is given by
\begin{equation}
  \lim_{\tau \to \infty} \frac{\rmd_{r_{ij}} \langle Q_1(\tau) \rangle}{\rmd_{r_{ij}} \langle Q_2(\tau) \rangle} = \frac{\sum_l c_l^{(1)} (Z_{li} - Z_{lj})}{\sum_l c_l^{(2)} (Z_{li} - Z_{lj})} \equiv \chi_{ij}^{(1)(2)}.
\end{equation}
Now we state that the ratio $\chi_{ij}^{(1)(2)}$ is independent of the transition rate $r_{ij}$. Please, see \cref{AppSec:proof_of_theorem2} for detailed proof.
\begin{theorem} \label{theorem: dchi_ij^(1)(2)/dr=0}
  The ratio $\chi_{ij}^{(1)(2)}$ is independent of the transition rate $r_{ij}$, i.e., $\rmd \chi_{ij}^{(1)(2)} / \rmd r_{ij} = 0$.
\end{theorem}

\cref{theorem: dchi_ij^(1)(2)/dr=0} shows that the ratio $\chi_{ij}^{(1)(2)}$ is independent of the transition rate $r_{ij}$, which means that the linear response of state-counting observables $\langle Q_1(\infty) \rangle$ and $\langle Q_2(\infty) \rangle$ to the perturbation on $r_{ij}$ are proportional to each other with a proportionality constant $\chi_{ij}^{(1)(2)}$ that does not depend on $r_{ij}$. Notice that observables $\langle Q_1(\infty) \rangle$ and $\langle Q_2(\infty) \rangle$ do not explicitly include the perturbed transition counts. More generally, the coefficients defining the observables may depend on the transition rates, as long as they do not depend on the perturbed rate $r_{ij}$. Therefore, the mutual-linearity relation can also apply to physically important thermodynamic observables, such as the local entropy production associated with edges other than the perturbed one. The mutual linearity can also be formulated as the linear dependence between $\langle Q_1(\infty) \rangle$ and $\langle Q_2(\infty) \rangle$ in the long time limit, which means that there exist constants $\chi_{ij}^{(1)(2)}$ and $\gamma_{ij}^{(1)(2)}$ such that
\begin{equation}
  \left.\langle Q_1(\infty) \rangle\right|_{r_{ij}} = \chi_{ij}^{(1)(2)} \left.\langle Q_2(\infty) \rangle\right|_{r_{ij}} + \gamma_{ij}^{(1)(2)},
  \label{eq: mutual linearity general}
\end{equation}
where $\gamma_{ij}^{(1)(2)}$ is a constant given by the intercept. The mutual linearity implies that the two observables $Q_1$ and $Q_2$ carry the same amount of information about the system's response to the perturbation on $r_{ij}$. This result recovers the mutual linearity for state-counting observables in \cite{bebon2026mutual} from the stochastic trajectory point of view.

\section{Frequency-domain mutual linearity in non-stationary relaxation dynamics} \label{sec: mutual linearity frequency}

The non-stationary relaxation here refers to the system that is relaxing from an initial distribution to the steady state with a time-independent transition rate matrix. The mutual linearity for current observables in the non-stationary regime has been derived in \cite{harunari2024mutual} using linear algebraic methods. Here, we provide an alternative proof based on the trajectory-level linear response theory and extend the mutual linearity to state-counting observables that do not explicitly include the perturbed transition counts to the non-stationary relaxation regime.

Consider the Laplace transform of the observable $\langle Q(\tau) \rangle$ with respect to time $\tau$:
\begin{equation}
    \hat{Q}(\omega) = \int_0^\infty e^{-\omega\tau} \langle Q(\tau) \rangle \rmd \tau,
\end{equation}
where $\omega$ is the Laplace variable. The Laplace transform $\hat{Q}(\omega)$ can be interpreted as the frequency domain representation of the observable $\langle Q(\tau) \rangle$. The linear response of $\hat{Q}(\omega)$ to a perturbation on the transition rate $r_{ij}$ can be computed as:
\begin{equation}
    \frac{\rmd \hat{Q}(\omega)}{\rmd r_{ij}} = \int_0^\infty e^{-\omega\tau} \frac{\rmd \langle Q(\tau) \rangle}{\rmd r_{ij}} \rmd \tau,
\end{equation}
where we change the order of differentiation and integration. Using the expression of $\rmd_{r_{ij}} \langle Q(\tau) \rangle$ obtained in \cref{eq: linear response}, we denote the response function as
\begin{equation}
  \mathcal{R}_{r_{ij}}(\tau, t) = \int_t^\tau \rmd t' \sum_l c_l p_j(t) [P(l, t'| i, t) - P(l, t'| j, t)],
\end{equation}
where $c_l = a_l + \sum_{k (\neq l) / j \to i} b_{kl} r_{kl}$ is a constant that depends on the coefficients of the state-counting observable. In this case, the finite-time linear response is given by $\rmd_{r_{ij}}\langle Q(\tau) \rangle = \int_0^\tau \mathcal{R}_{r_{ij}}(\tau, t) \rmd t$. Since the system is time-homogeneous, i.e., the original transition rate $r_{ij}$ is time-independent, the transition probabilities only depend on the time difference $t'-t$. Therefore, we can change the variable from $t'$ to $s = t'-t$ and rewrite the response function as
\begin{equation}
  \mathcal{R}_{r_{ij}}(\tau, t) = p_j(t) \sum_l c_l \int_0^{\tau-t} \rmd s ~ [P(l, s| i) - P(l, s| j)].
  \label{eq: response function state-counting}
\end{equation}
Note that the response function $\mathcal{R}_{r_{ij}}(\tau, t)$ depends on both $\tau$ and $t$ since $\boldsymbol{p}(t)$ is non-stationary.

The linear response of $\hat{Q}(\omega)$ to the perturbation on $r_{ij}$ can be expressed as the double integral of the response function $\mathcal{R}_{r_{ij}}(\tau, t)$:
\begin{subequations}
\begin{align}
  \frac{\rmd \hat{Q}(\omega)}{\rmd r_{ij}} &= \int_0^\infty \rmd \tau ~ e^{-\omega\tau} \int_0^\tau \rmd t ~ \mathcal{R}_{r_{ij}}(\tau, t) \\
  &= \int_0^\infty \rmd t \int_t^\infty \rmd \tau ~ e^{-\omega\tau} \mathcal{R}_{r_{ij}}(\tau, t),
\end{align}
\end{subequations}
where the second line swaps the order of integration. The Laplace transform of the response function $\mathcal{R}_{r_{ij}}(\tau, t)$ with respect to $\tau$ is given by
\begin{subequations}
\begin{align}
  \hat{\mathcal{R}}_{r_{ij}}(\omega, t) \equiv{}& \int_t^\infty \rmd \tau ~ e^{-\omega\tau} \mathcal{R}_{r_{ij}}(\tau, t) \\
  ={}& e^{-\omega t} p_j(t) \sum_l c_l \int_0^\infty \rmd u ~ e^{-\omega u} \\
  &\times \int_0^u ~ \rmd s [P(l, s| i) - P(l, s| j)] \\
  ={}& \frac{1}{\omega} e^{-\omega t} p_j(t) \sum_l c_l \left[ \hat{P}(l, \omega| i) - \hat{P}(l, \omega| j) \right],
\end{align}
\end{subequations}
where we change the variable from $\tau$ to $u = \tau-t$ and use the definition of the Laplace transform $\hat{P}(l, \omega| i) = \int_0^\infty e^{-\omega s} P(l, s| i) \rmd s$. The third equation comes from the fact that $\int_0^u \rmd s [P(l, s| i) - P(l, s| j)]$ is the convolution of $P(l, s| i) - P(l, s| j)$ and the Heaviside step function $\Theta(u-s)$, whose Laplace transform is given by $[\hat{P}(l, \omega| i) - \hat{P}(l, \omega| j)] / \omega$.
Therefore, the linear response of $\hat{Q}(\omega)$ to the perturbation on $r_{ij}$ can be expressed as
\begin{subequations}
\begin{align}
  \frac{\rmd \hat{Q}(\omega)}{\rmd r_{ij}} &= \int_0^\infty \rmd t ~ \hat{\mathcal{R}}_{r_{ij}}(\omega, t) \\
  &= \frac{1}{\omega} \sum_l c_l \left[ \hat{P}(l, \omega| i) - \hat{P}(l, \omega| j) \right] \int_0^\infty e^{-\omega t} p_j(t) \rmd t \\
  &= \frac{1}{\omega} \hat{p}_j(\omega) \sum_l c_l \left[ \hat{P}(l, \omega| i) - \hat{P}(l, \omega| j) \right],
\end{align}
\end{subequations}
where $\hat{p}_j(\omega) = \int_0^\infty e^{-\omega t} p_j(t) \rmd t$ is the Laplace transform of $p_j(t)$.

The ratio of the linear response of $\hat{Q}_1(\omega)$ and $\hat{Q}_2(\omega)$ to the perturbation on $r_{ij}$ is given by
\begin{equation}
  \frac{\rmd_{r_{ij}} \hat{Q}_1(\omega)}{\rmd_{r_{ij}} \hat{Q}_2(\omega)} = \frac{\sum_l c_l^{(1)} \left[ \hat{P}(l, \omega| i) - \hat{P}(l, \omega| j) \right]}{\sum_l c_l^{(2)} \left[ \hat{P}(l, \omega| i) - \hat{P}(l, \omega| j) \right]} \equiv \hat{\chi}_{ij}^{(1)(2)}(\omega).
\end{equation}
We now state that the ratio $\hat{\chi}_{ij}^{(1)(2)}(\omega)$ is independent of the transition rate $r_{ij}$. The detailed proof can be found in \cref{AppSec:proof_of_theorem3}.
\begin{theorem} \label{theorem: dchi_ij^(1)(2)(omega)/dr=0}
  For $\operatorname{Re}\omega > 0$, the ratio $\hat{\chi}_{ij}^{(1)(2)}(\omega)$ is independent of the transition rate $r_{ij}$, i.e., $\rmd \hat{\chi}_{ij}^{(1)(2)}(\omega) / \rmd r_{ij} = 0$, for observables that exclude transitions from $j$ to $i$.
\end{theorem}

\cref{theorem: dchi_ij^(1)(2)(omega)/dr=0} shows that the ratio $\hat{\chi}_{ij}^{(1)(2)}(\omega)$ is independent of the transition rate $r_{ij}$, which means that the linear response of $\hat{Q}_1(\omega)$ and $\hat{Q}_2(\omega)$ to the perturbation on $r_{ij}$ are proportional to each other with a proportionality constant $\hat{\chi}_{ij}^{(1)(2)}(\omega)$ that does not depend on $r_{ij}$. The mutual linearity can also be formulated as the linear dependence between $\hat{Q}_1(\omega)$ and $\hat{Q}_2(\omega)$ in the frequency domain. For $\operatorname{Re}\omega > 0$ and $\rmd_{r_{ij}} \hat{Q}_2(\omega) \neq 0$, there exist constants $\hat{\chi}_{ij}^{(1)(2)}(\omega)$ and $\hat{\gamma}_{ij}^{(1)(2)}(\omega)$ such that
\begin{equation}
  \left.\hat{Q}_1(\omega) \right|_{r_{ij}} = \hat{\chi}_{ij}^{(1)(2)}(\omega) \left.\hat{Q}_2(\omega) \right|_{r_{ij}} + \hat{\gamma}_{ij}^{(1)(2)}(\omega),
  \label{eq: mutual linearity non-stationary}
\end{equation}
where $\hat{\gamma}_{ij}^{(1)(2)}(\omega)$ is the intercept. The mutual linearity implies that the spectrum of two observables $Q_1$ and $Q_2$ contains the same amount of information about the system's response to the perturbation on $r_{ij}$. This result recovers the mutual linearity for current observables in the non-stationary regime in \cite{harunari2024mutual} and extends the linear relation from current to state-counting observables that do not explicitly include the perturbed transition counts to the non-stationary regime from the stochastic trajectory point of view. The stationary mutual linearity can be recovered by taking the limit $\omega \to 0$ in \cref{eq: mutual linearity non-stationary}. The frequency-domain result shows that mutual linearity is not restricted to steady states, but is a frequency-resolved property of transient dynamics.

\section{Numerical Illustrations} \label{sec: numerical}

In this section, we provide numerical verification of the frequency-domain mutual linearity \cref{eq: mutual linearity non-stationary} derived in \cref{sec: mutual linearity frequency}. The validation is carried out using stochastic trajectory simulations based on the Gillespie algorithm \cite{gillespie1977exact}, which allows for a direct evaluation of Laplace-domain observables without relying on matrix-based analytical solutions. We consider the simple exclusion processes (SEP) and demonstrate that the predicted linear relation between Laplace-transformed observables holds robustly across a wide range of dynamical regimes.

We consider an interacting lattice transport model in the form of an open one-dimensional simple exclusion process with $N$ sites, coupled to particle reservoirs and a thermal bath at temperature $T$. Each site $i$ can be either empty or occupied by at most one particle, so that the configuration space consists of $2^N$ states labeled by
\begin{equation}
  x = (\alpha_1,\alpha_2,\dots,\alpha_N), \qquad \alpha_i \in \{0,1\}.
\end{equation}

Each site is assigned an energy $E_i$, and particle transport between neighboring sites involves overcoming the transition state energy $B_{i,i+1} = B_{i+1,i}$. The system is coupled to left and right particle reservoirs with chemical potentials $\mu_L$ and $\mu_R$, respectively. The entire system is in contact with a thermal bath at inverse temperature $\beta = 1/(k_B T)$.

The dynamics are described as a Markov jump process with transition rates satisfying an Arrhenius form. For a transition $x \to x'$ involving a particle hop or exchange with a reservoir, the transition rate is taken as
\begin{equation}
  k_{x',x} = \exp\left[-\beta\left( B_{x', x} - E_x \right)\right],
\end{equation}
where $E_x$ is the energy of configuration $x$, and $B_{x', x}$ is the corresponding transition-state energy. For bulk hopping between neighboring sites $i$ and $i+1$, the transition corresponds to moving a particle across a bond. The energy change is determined by the site energies, leading to rates of the form
\begin{equation}
  k_{i+1, i} = \exp\left[-\beta\left(B_{i+1,i} - E_i  \right)\right],
\end{equation}
and similarly for the reverse process. At the boundaries, the system exchanges particles with reservoirs characterized by chemical potentials $\mu_L$ and $\mu_R$. At the left boundary, particle injection and extraction are described by
\begin{align}
  k^{\mathrm{in}}_L &= \exp \left[ -\beta \left( B_L - \mu_L \right) \right], \\
  k^{\mathrm{out}}_L &= \exp \left[ -\beta \left(B_L - E_1 \right) \right],
\end{align}
where $B_L$ is the barrier associated with the left boundary. Similarly, at the right boundary, we have
\begin{align}
  k^{\mathrm{in}}_R &= \exp \left[ -\beta \left( B_R - \mu_R \right) \right], \\
  k^{\mathrm{out}}_R &= \exp \left[ -\beta \left( B_R - E_N \right) \right],
\end{align}
with $B_R$ the barrier at the right boundary.

We introduce a perturbation by modifying a single transition rate, namely the injection of a particle from the left reservoir into the completely empty configuration,
\begin{equation}
00\cdots0 \;\to\; 10\cdots0,
\end{equation}
and denote this rate by $\lambda$. Physically, this corresponds to tuning the injection rate at the left boundary in the low-density regime, while all other rates remain determined by the underlying energy landscape.

We consider four time-integrated observables. The first is the net particle current flowing into the right reservoir,
\begin{equation}
    Q_1(\tau) = \int_0^\tau J_R(t)\,\rmd t,
\end{equation}
where $J_R(t)$ counts particle transfers across the right boundary. The second and third observables are the dwelling times in the fully occupied configuration and the empty configuration, respectively,
\begin{equation}
    Q_2(\tau) = \int_0^\tau \rmd \tau_{11\cdots1}, \quad Q_3(\tau) = \int_0^\tau \rmd \tau_{00\cdots0}.
\end{equation}
To illustrate that thermodynamic observables are not excluded from the mutual-linearity structure, we also consider a local entropy-production observable associated with the first bulk bond. Let $x^{1\to2}$ denote the configuration obtained from $x$ by moving one particle from site $1$ to site $2$, and let $x^{2\to1}$ denote the reverse move. We define
\begin{align}
    Q_4(\tau) &= \sum_{x:\alpha_1=1,\alpha_2=0} \ln \frac{k_{x^{1\to2},x}}{k_{x,x^{1\to2}}} n_{x^{1\to2},x}(\tau) \nonumber\\
    &\quad + \sum_{x:\alpha_1=0,\alpha_2=1} \ln \frac{k_{x^{2\to1},x}}{k_{x,x^{2\to1}}} n_{x^{2\to1},x}(\tau).
\end{align}
This observable is the medium entropy production accumulated along the microscopic transitions crossing the first bulk bond. Importantly, these transitions do not include the perturbed channel $00\cdots0 \to 10\cdots0$. Therefore, $Q_4$ falls within the class of additive observables covered by the mutual-linearity theorem. By contrast, a total entropy-production observable involving all transition channels, or an observable whose jump weights themselves depend on the perturbed parameter, is not expected to satisfy the single-edge mutual-linearity relation in general.

We define the Laplace-transformed observables
\begin{equation}
    \hat{Q}_i(\omega) = \int_0^\infty e^{-\omega\tau} \langle Q_i(\tau) \rangle \rmd \tau, \qquad i = 1,2,3,4.
\end{equation}

We numerically simulate a $(N = 3)$-site model and a $(N = 8)$-site model, which correspond to $8$-state and $256$-state Markov networks, respectively. The results are shown in \cref{fig: SEP,fig: SEP entropy}. For each fixed $\omega$, the parametric plots collapse onto straight lines as $\lambda$ is varied, in agreement with the theoretical prediction. In particular, \cref{fig: SEP entropy} shows that the local entropy-production observable $\hat{Q}_4(\omega)$ is mutually linear with both $\hat{Q}_2(\omega)$ and $\hat{Q}_3(\omega)$. This demonstrates that physically important thermodynamic observables can satisfy mutual linearity, provided that the observable does not contain the perturbed transition channel.

\begin{figure}[htbp]
  \centering
  \includegraphics[width=1.0\linewidth]{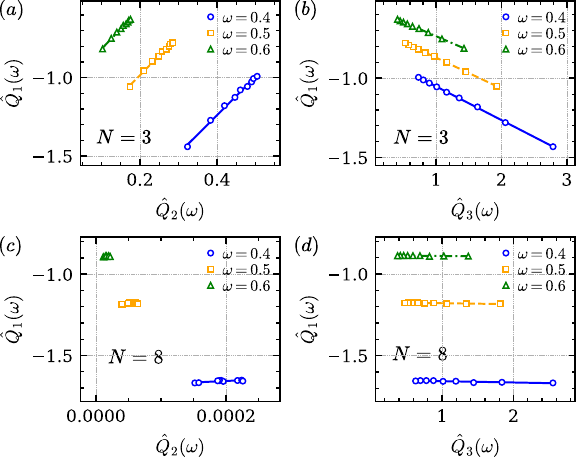}
  \caption{
  \textbf{Parametric plot of $\hat{Q}_1(\omega)$ versus $\hat{Q}_2(\omega)$ and $\hat{Q}_3(\omega)$ for the SEP model.}
  All simulations are initialized from the same delta distribution over the empty configuration. The inverse temperature is set to $\beta = 0.1$, the chemical potentials are $\mu_L = 2.0$ and $\mu_R = 0.0$, and the energy landscape is given by $E_i = i/N$ and $B_{i,i+1} = 1.5$ for all $i$. The perturbed rate $\lambda$ is varied from $0.1$ to $5.0$. Each data point is averaged over $50000$ trajectories with time length $\tau = 500.0$.}
  \label{fig: SEP}
\end{figure}

\begin{figure}[htbp]
    \centering
    \includegraphics[width=1.0\linewidth]{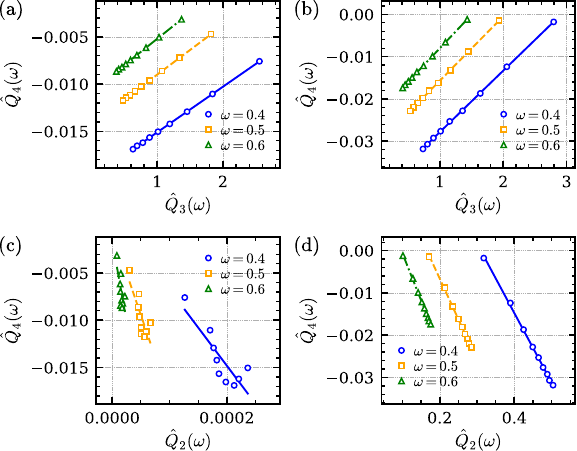}
    \caption{
    \textbf{Parametric plot of the local entropy production $\hat{Q}_4(\omega)$ versus $\hat{Q}_2(\omega)$ and $\hat{Q}_3(\omega)$ for the SEP model.}
    The observable $Q_4$ is the medium entropy production accumulated on the first bulk bond $1\leftrightarrow2$, which does not include the perturbed transition $00\cdots0 \to 10\cdots0$. All simulations are initialized from the same delta distribution over the empty configuration. The inverse temperature is set to $\beta = 0.1$, the chemical potentials are $\mu_L = 2.0$ and $\mu_R = 0.0$, and the energy landscape is given by $E_i = i/N$ and $B_{i,i+1} = 1.5$ for all $i$. The perturbed rate $\lambda$ is varied from $0.1$ to $5.0$. Each data point is averaged over $50000$ trajectories with time length $\tau = 500.0$.}
    \label{fig: SEP entropy}
\end{figure}

\section{Conclusion and Discussion}

In this work, we have developed a trajectory-level framework to understand mutual linearity in Markov jump processes. By expressing the linear response in terms of correlations with martingale noise via the Doob-Meyer decomposition, we showed that mutual linearity arises from a simple multiplicative structure of the response kernel associated with the transition probabilities. This provides a transparent trajectory-level interpretation of a result that was previously derived using linear algebraic methods \cite{harunari2024mutual,bebon2026mutual}.

Our approach reveals that mutual linearity is not merely a consequence of specific algebraic properties of the generator, but rather reflects a general dynamical structure of autonomous Markov jump processes. In particular, a local perturbation of a single transition channel propagates through the system along the same set of transition probabilities, leading to proportional responses of different observables. This mechanism explains why distinct observables share the same response structure.

The result established in this work applies to finite-state Markov jump processes under single-edge perturbations. In general, perturbations that affect multiple transition channels do not necessarily preserve mutual linearity, because the corresponding response is a sum of several edge-resolved response kernels with different dynamical structures. An interesting and physically important exception may arise in chemical reaction networks. Although such systems have countably infinite state spaces, their transition structure is highly regular: perturbing the rate constant of a chemical reaction changes all microscopic transition rates associated with the same stoichiometric jump direction. This lattice-like structure may impose additional constraints that restore a form of mutual linearity beyond the finite-state single-edge setting considered here. This expectation is also supported by the recent proof of mutual linearity for Langevin dynamics \cite{zheng2026mutual}, since stochastic chemical reaction networks are often be approximated by Langevin equations in suitable large-volume limits. We therefore conjecture that mutual linearity may also hold for chemical reaction networks described exactly at the master-equation level. A rigorous proof of this conjecture is left for future work.

Furthermore, the trajectory-based formulation naturally extends mutual linearity beyond steady states. By analyzing the Laplace-transformed response, we demonstrated that the linear dependence between observables persists in the frequency domain for non-stationary dynamics. This shows that mutual linearity is a frequency-resolved property of the response, characterizing both transient and steady-state behavior within a unified framework.

Finally, since similar trajectory and martingale techniques are well established for diffusion processes \cite{dieball2023direct} and open quantum systems \cite{kwon2025unified}, our results suggest a promising route to generalizing mutual linearity to continuous systems and quantum systems. Exploring such extensions, as well as their implications for fluctuation-response relations and uncertainty bounds, would be a valuable direction for future work.

\section{Acknowledgements}

This work is supported by the U.S. National Science Foundation under Grant No. DMR-2145256 and Alfred P. Sloan Foundation Matter-to-Life Theory Award under Grant No. G-2025-25194.

\section{Data availability}

The data that support the findings of this article are generated by numerical simulation codes that are openly available at \cite{data}.

\appendix
\section{Proof of \texorpdfstring{\Cref{theorem: dchi_ij^kl/dr=0}}{Theorem 1}} \label[appendix]{AppSec:proof_of_theorem1}

\noindent \textbf{\cref{theorem: dchi_ij^kl/dr=0}.} For observables that excludes the transition from $j$ to $i$, the ratio $\chi_{ij}^{kl}$ is independent of the transition rate $r_{ij}$, i.e., $\rmd \chi_{ij}^{kl} / \rmd r_{ij} = 0$.

\begin{proof}
  The fundamental matrix $Z$ is the pseudoinverse of the transition rate matrix $R$, which satisfies the following properties:
  \begin{align}
    RZ &= ZR = \boldsymbol{\pi}\boldsymbol{1}^\top - I,
    \label{eq: pseudoinverse}
  \end{align}
  where $I$ is the identity matrix and $\boldsymbol{0}$ is the zero vector. This can be obtained by using $Re^{Rt} = e^{Rt}R = \frac{\rmd}{\rmd t} e^{Rt}$, $R\boldsymbol{\pi} = 0$, $\boldsymbol{1}^\top R = 0$, and $\lim_{t \to \infty} e^{Rt} = \boldsymbol{\pi}\boldsymbol{1}^\top$:
  \begin{subequations}
  \begin{align}
    RZ &= \int_0^\infty R e^{Rt} \rmd t - \int_0^\infty R \boldsymbol{\pi}\boldsymbol{1}^\top \rmd t \\
    &= \int_0^\infty \frac{\rmd}{\rmd t} e^{Rt} \rmd t \\
    &= \boldsymbol{\pi}\boldsymbol{1}^\top - I, \\
    ZR &= \int_0^\infty e^{Rt} R \rmd t - \int_0^\infty \boldsymbol{\pi}\boldsymbol{1}^\top R \rmd t \\
    &= \int_0^\infty \frac{\rmd}{\rmd t} e^{Rt} \rmd t \\
    &= \boldsymbol{\pi}\boldsymbol{1}^\top - I.
  \end{align}
  \end{subequations}
  The fundamental matrix $Z$ also satisfies the orthogonality condition $Z\boldsymbol{\pi} = \boldsymbol{0}$ and $\boldsymbol{1}^\top Z = \boldsymbol{0}^\top$. This can be obtained by using $e^{Rt}\boldsymbol{\pi} = \boldsymbol{\pi}$ and $\boldsymbol{1}^\top e^{Rt} = \boldsymbol{1}^\top$.

  We first compute the derivative of $Z$ with respect to $r_{ij}$. Taking the derivative of the pseudoinverse property \cref{eq: pseudoinverse} with respect to $r_{ij}$ leads to
  \begin{equation}
    \frac{\rmd R}{\rmd r_{ij}} Z + R \frac{\rmd Z}{\rmd r_{ij}} = \frac{\rmd \boldsymbol{\pi}}{\rmd r_{ij}} \boldsymbol{1}^\top.
  \end{equation}
  Multiplying both sides of the above equation by $Z$ from the left and using the pseudoinverse property again, we have
  \begin{equation}
    Z \frac{\rmd R}{\rmd r_{ij}} Z + (\boldsymbol{\pi}\boldsymbol{1}^\top - I) \frac{\rmd Z}{\rmd r_{ij}} = Z \frac{\rmd \boldsymbol{\pi}}{\rmd r_{ij}} \boldsymbol{1}^\top.
  \end{equation}
  Rearranging the above equation gives
  \begin{subequations}
  \begin{align}
    - (\boldsymbol{\pi}\boldsymbol{1}^\top - I) \frac{\rmd Z}{\rmd r_{ij}} &= - \boldsymbol{\pi}\frac{\rmd (\boldsymbol{1}^\top Z)}{\rmd r_{ij}} + \frac{\rmd Z}{\rmd r_{ij}} \\
    &= \frac{\rmd Z}{\rmd r_{ij}} \\
    &= Z \frac{\rmd R}{\rmd r_{ij}} Z - Z \frac{\rmd \boldsymbol{\pi}}{\rmd r_{ij}} \boldsymbol{1}^\top,
  \end{align}
  \end{subequations}
  where we use $\boldsymbol{1}^\top Z = \boldsymbol{0}^\top$ to obtain the second line.

  Now we calculate the derivative of $Z_{ki}-Z_{kj}$ and $Z_{li}-Z_{lj}$ with respect to $r_{ij}$. Let $E_{ij}$ be the matrix with all elements equal to zero except for the element at the $i$-th row and $j$-th column, which is equal to one. We have $\frac{\rmd R}{\rmd r_{ij}} = E_{ij} - E_{jj}$. Let $\boldsymbol{e}_k$ be the column vector with all elements equal to zero except for the $k$-th element, which is equal to one. Let $\boldsymbol{v} = Z(\boldsymbol{e}_i - \boldsymbol{e}_j)$, we have $\boldsymbol{v}_k = Z_{ki}-Z_{kj}$. The derivative of $\boldsymbol{v}$ with respect to $r_{ij}$ is given by
  \begin{subequations}
  \begin{align}
    \frac{\rmd \boldsymbol{v}}{\rmd r_{ij}} ={}& \frac{\rmd Z}{\rmd r_{ij}} (\boldsymbol{e}_i - \boldsymbol{e}_j) \\
    ={}& Z (E_{ij} - E_{jj}) Z (\boldsymbol{e}_i - \boldsymbol{e}_j) \nonumber \\
    &- Z \frac{\rmd \boldsymbol{\pi}}{\rmd r_{ij}} \boldsymbol{1}^\top (\boldsymbol{e}_i - \boldsymbol{e}_j) \\
    ={}& Z (E_{ij} - E_{jj}) Z (\boldsymbol{e}_i - \boldsymbol{e}_j) \label{eq: normalization} \\
    ={}& Z (E_{ij} - E_{jj}) \boldsymbol{v} \\
    ={}& Z \boldsymbol{v}_j (\boldsymbol{e}_i - \boldsymbol{e}_j) \label{eq: v_j} \\
    ={}& \boldsymbol{v}_j \boldsymbol{v},
  \end{align}
  \end{subequations}
  where we use $\boldsymbol{1}^\top (\boldsymbol{e}_i - \boldsymbol{e}_j) = 0$ to obtain \cref{eq: normalization}. \cref{eq: v_j} is obtained by noting that $(E_{ij} \boldsymbol{v})_k = \boldsymbol{v}_j \delta_{ik}$ and $(E_{jj} \boldsymbol{v})_k = \boldsymbol{v}_j \delta_{jk}$, which leads to $Z (E_{ij} - E_{jj}) \boldsymbol{v} = Z \boldsymbol{v}_j (\boldsymbol{e}_i - \boldsymbol{e}_j)$. The last line is obtained by noting that $\boldsymbol{v}_j$ is a scalar.

  The above equation implies that
  \begin{subequations}
  \begin{align}
    \frac{\rmd (Z_{ki}-Z_{kj})}{\rmd r_{ij}} &= \frac{\rmd \boldsymbol{v}_k}{\rmd r_{ij}} = \boldsymbol{v}_j \boldsymbol{v}_k, \\
    \frac{\rmd (Z_{li}-Z_{lj})}{\rmd r_{ij}} &= \frac{\rmd \boldsymbol{v}_l}{\rmd r_{ij}} = \boldsymbol{v}_j \boldsymbol{v}_l.
  \end{align}
  \end{subequations}
  Therefore, we have
  \begin{subequations}
  \begin{align}
    \frac{\rmd \chi_{ij}^{kl}}{\rmd r_{ij}} &= \frac{\rmd}{\rmd r_{ij}} \frac{\boldsymbol{v}_k}{\boldsymbol{v}_l} \\
    &= \frac{\frac{\rmd \boldsymbol{v}_k}{\rmd r_{ij}} \boldsymbol{v}_l - \boldsymbol{v}_k \frac{\rmd \boldsymbol{v}_l}{\rmd r_{ij}}}{\boldsymbol{v}_l^2} \\
    &= \frac{(\boldsymbol{v}_j\boldsymbol{v}_k)\boldsymbol{v}_l - \boldsymbol{v}_k (\boldsymbol{v}_j\boldsymbol{v}_l)}{\boldsymbol{v}_l^2} \\
    &= 0.
  \end{align}
  \end{subequations}
\end{proof}

\section{Proof of \texorpdfstring{\Cref{theorem: dchi_ij^(1)(2)/dr=0}}{Theorem 2}} \label[appendix]{AppSec:proof_of_theorem2}

\noindent \textbf{\cref{theorem: dchi_ij^(1)(2)/dr=0}.} The ratio $\chi_{ij}^{(1)(2)}$ is independent of the transition rate $r_{ij}$, i.e., $\rmd \chi_{ij}^{(1)(2)} / \rmd r_{ij} = 0$.
  
\begin{proof}
  The derivative of $\chi_{ij}^{(1)(2)}$ with respect to $r_{ij}$ is given by
  \begin{subequations}
  \begin{align}
    \frac{\rmd \chi_{ij}^{(1)(2)}}{\rmd r_{ij}} &= \frac{\rmd}{\rmd r_{ij}} \frac{\sum_l c_l^{(1)}\boldsymbol{v}_l}{\sum_l c_l^{(2)}\boldsymbol{v}_l} \\
    ={}& \frac{\sum_l c_l^{(1)} \frac{\rmd \boldsymbol{v}_l}{\rmd r_{ij}} \sum_l c_l^{(2)} \boldsymbol{v}_l - \sum_l c_l^{(1)} \boldsymbol{v}_l \sum_l c_l^{(2)} \frac{\rmd \boldsymbol{v}_l}{\rmd r_{ij}}}{\left[ \sum_l c_l^{(2)} \boldsymbol{v}_l \right]^2} \\
    ={}& \frac{\sum_l c_l^{(1)} \boldsymbol{v}_j \boldsymbol{v}_l \sum_l c_l^{(2)} \boldsymbol{v}_l - \sum_l c_l^{(1)} \boldsymbol{v}_l \sum_l c_l^{(2)} \boldsymbol{v}_j \boldsymbol{v}_l}{\left[ \sum_l c_l^{(2)} \boldsymbol{v}_l \right]^2} \\
    ={}& 0,
  \end{align}
  \end{subequations}
  where we have used the result $\frac{\rmd \boldsymbol{v}_l}{\rmd r_{ij}} = \boldsymbol{v}_j \boldsymbol{v}_l$ with $\boldsymbol{v}_l \equiv (Z_{li} - Z_{lj})$ obtained in the proof of \cref{theorem: dchi_ij^kl/dr=0}.
\end{proof}

\section{Proof of \texorpdfstring{\Cref{theorem: dchi_ij^(1)(2)(omega)/dr=0}}{Theorem 3}} \label[appendix]{AppSec:proof_of_theorem3}

\noindent \textbf{\cref{theorem: dchi_ij^(1)(2)(omega)/dr=0}.} For $\operatorname{Re}\omega > 0$, the ratio $\hat{\chi}_{ij}^{(1)(2)}(\omega)$ is independent of the transition rate $r_{ij}$, i.e., $\rmd \hat{\chi}_{ij}^{(1)(2)}(\omega) / \rmd r_{ij} = 0$, for observables that exclude transitions from $j$ to $i$.
  
\begin{proof}
  Let $\hat{P}(\omega) = \int_0^\infty e^{-\omega t} e^{Rt} \rmd t$ be the Laplace transform of the transition probability matrix $P(t) = e^{Rt}$. The ratio can be rewritten as
  \begin{equation}
  \hat{\chi}_{ij}^{(1)(2)}(\omega) = \frac{\sum_l c_l^{(1)} \left[ \hat{P}_{li}(\omega) - \hat{P}_{lj}(\omega) \right]}{\sum_l c_l^{(2)} \left[ \hat{P}_{li}(\omega) - \hat{P}_{lj}(\omega) \right]}.
  \end{equation}
  The matrix $\hat{P}(\omega)$ satisfies the following equation:
  \begin{subequations}
  \begin{align}
    R \hat{P}(\omega) &= \int_0^\infty e^{-\omega t} Re^{Rt} \rmd t \\
    &= \int_0^\infty e^{-\omega t} \left( \frac{\rmd}{\rmd t} e^{Rt} \right) \rmd t \\
    &= \left[ e^{-\omega t} e^{Rt} \right]_0^\infty + \omega \int_0^\infty e^{-\omega t} e^{Rt} \rmd t \\
    &= - I + \omega \hat{P}(\omega),
  \end{align}
  \end{subequations}
  where we use the fact that $\frac{\rmd}{\rmd t} e^{Rt} = R e^{Rt}$ and $P(0) = I$ to obtain the third line. The rearrangement gives $\hat{P}(\omega) = (\omega I - R)^{-1}$.

  Taking the derivative of $(\omega I - R) \hat{P}(\omega) = I$ with respect to $r_{ij}$, we have
  \begin{equation}
    -\frac{\rmd R}{\rmd r_{ij}} \hat{P} + (\omega I - R) \frac{\rmd \hat{P}}{\rmd r_{ij}} = 0.
  \end{equation}
  Rearranging the above equation gives
  \begin{subequations}
  \begin{align}
    \frac{\rmd \hat{P}}{\rmd r_{ij}} &= (\omega I - R)^{-1} \frac{\rmd R}{\rmd r_{ij}} \hat{P} \\
    &= \hat{P} \frac{\rmd R}{\rmd r_{ij}} \hat{P} \\
    &= \hat{P} (E_{ij} - E_{jj}) \hat{P}
  \end{align}
  \end{subequations}
  where we use $\frac{\rmd R}{\rmd r_{ij}} = E_{ij} - E_{jj}$.

  Then, let $\hat{\boldsymbol{v}} = \hat{P} (\boldsymbol{e}_i - \boldsymbol{e}_j)$, we have $\hat{\boldsymbol{v}}_l = \hat{P}_{li} - \hat{P}_{lj}$. The derivative of $\hat{\boldsymbol{v}}$ with respect to $r_{ij}$ is given by
  \begin{equation}
    \frac{\rmd \hat{\boldsymbol{v}}}{\rmd r_{ij}} = \hat{P} (E_{ij} - E_{jj}) \hat{P} (\boldsymbol{e}_i - \boldsymbol{e}_j) = \hat{\boldsymbol{v}}_j \hat{\boldsymbol{v}},
  \end{equation}
  where we have used the result $(E_{ij} - E_{jj}) \hat{P} = \boldsymbol{v}_j (\boldsymbol{e}_i - \boldsymbol{e}_j)$ obtained by noting that $(E_{ij} \hat{\boldsymbol{v}})_l = \hat{\boldsymbol{v}}_j \delta_{il}$ and $(E_{jj} \hat{\boldsymbol{v}})_l = \hat{\boldsymbol{v}}_j \delta_{jl}$. Therefore, the derivative of the ratio $\hat{\chi}_{ij}^{(1)(2)}$ with respect to $r_{ij}$ is zero:
  \begin{subequations}
  \begin{align}
    \frac{\rmd \hat{\chi}_{ij}^{(1)(2)}}{\rmd r_{ij}} &= \frac{\rmd}{\rmd r_{ij}} \frac{\sum_l c_l^{(1)} \left( \hat{P}_{li} - \hat{P}_{lj} \right)}{\sum_l c_l^{(2)} \left( \hat{P}_{li} - \hat{P}_{lj} \right)} \\
    ={}& \frac{\sum_l c_l^{(1)} \hat{\boldsymbol{v}}_j \hat{\boldsymbol{v}}_l \sum_l c_l^{(2)} \hat{\boldsymbol{v}}_l - \sum_l c_l^{(1)} \hat{\boldsymbol{v}}_l \sum_l c_l^{(2)} \hat{\boldsymbol{v}}_j \hat{\boldsymbol{v}}_l}{\left[ \sum_l c_l^{(2)} \hat{\boldsymbol{v}}_l \right]^2} \\
    ={}& 0.
  \end{align}
  \end{subequations}
  It is worth noting that $\hat{P}(\omega)$ diverges as $\omega \to 0$ due to the presence of the zero eigenvalue of $R$. However, the difference $\hat{P}_{li}(\omega) - \hat{P}_{lj}(\omega)$ remains finite as $\omega \to 0$ and converges to $Z_{li} - Z_{lj}$, which is consistent with the result in the stationary regime.
\end{proof}

\bibliography{manuscript}

\begin{thebibliography}{49}%
\makeatletter
\providecommand \@ifxundefined [1]{%
 \@ifx{#1\undefined}
}%
\providecommand \@ifnum [1]{%
 \ifnum #1\expandafter \@firstoftwo
 \else \expandafter \@secondoftwo
 \fi
}%
\providecommand \@ifx [1]{%
 \ifx #1\expandafter \@firstoftwo
 \else \expandafter \@secondoftwo
 \fi
}%
\providecommand \natexlab [1]{#1}%
\providecommand \enquote  [1]{``#1''}%
\providecommand \bibnamefont  [1]{#1}%
\providecommand \bibfnamefont [1]{#1}%
\providecommand \citenamefont [1]{#1}%
\providecommand \href@noop [0]{\@secondoftwo}%
\providecommand \href [0]{\begingroup \@sanitize@url \@href}%
\providecommand \@href[1]{\@@startlink{#1}\@@href}%
\providecommand \@@href[1]{\endgroup#1\@@endlink}%
\providecommand \@sanitize@url [0]{\catcode `\\12\catcode `\$12\catcode `\&12\catcode `\#12\catcode `\^12\catcode `\_12\catcode `\%12\relax}%
\providecommand \@@startlink[1]{}%
\providecommand \@@endlink[0]{}%
\providecommand \url  [0]{\begingroup\@sanitize@url \@url }%
\providecommand \@url [1]{\endgroup\@href {#1}{\urlprefix }}%
\providecommand \urlprefix  [0]{URL }%
\providecommand \Eprint [0]{\href }%
\providecommand \doibase [0]{https://doi.org/}%
\providecommand \selectlanguage [0]{\@gobble}%
\providecommand \bibinfo  [0]{\@secondoftwo}%
\providecommand \bibfield  [0]{\@secondoftwo}%
\providecommand \translation [1]{[#1]}%
\providecommand \BibitemOpen [0]{}%
\providecommand \bibitemStop [0]{}%
\providecommand \bibitemNoStop [0]{.\EOS\space}%
\providecommand \EOS [0]{\spacefactor3000\relax}%
\providecommand \BibitemShut  [1]{\csname bibitem#1\endcsname}%
\let\auto@bib@innerbib\@empty
\bibitem [{\citenamefont {Mora}(2015)}]{mora2015physical}%
  \BibitemOpen
  \bibfield  {author} {\bibinfo {author} {\bibfnamefont {T.}~\bibnamefont {Mora}},\ }\bibfield  {title} {\bibinfo {title} {Physical limit to concentration sensing amid spurious ligands},\ }\href@noop {} {\bibfield  {journal} {\bibinfo  {journal} {Physical review letters}\ }\textbf {\bibinfo {volume} {115}},\ \bibinfo {pages} {038102} (\bibinfo {year} {2015})}\BibitemShut {NoStop}%
\bibitem [{\citenamefont {Bialek}\ and\ \citenamefont {Setayeshgar}(2005)}]{bialek2005physical}%
  \BibitemOpen
  \bibfield  {author} {\bibinfo {author} {\bibfnamefont {W.}~\bibnamefont {Bialek}}\ and\ \bibinfo {author} {\bibfnamefont {S.}~\bibnamefont {Setayeshgar}},\ }\bibfield  {title} {\bibinfo {title} {Physical limits to biochemical signaling},\ }\href@noop {} {\bibfield  {journal} {\bibinfo  {journal} {Proceedings of the National Academy of Sciences}\ }\textbf {\bibinfo {volume} {102}},\ \bibinfo {pages} {10040} (\bibinfo {year} {2005})}\BibitemShut {NoStop}%
\bibitem [{\citenamefont {Hartich}\ \emph {et~al.}(2016)\citenamefont {Hartich}, \citenamefont {Barato},\ and\ \citenamefont {Seifert}}]{hartich2016sensory}%
  \BibitemOpen
  \bibfield  {author} {\bibinfo {author} {\bibfnamefont {D.}~\bibnamefont {Hartich}}, \bibinfo {author} {\bibfnamefont {A.~C.}\ \bibnamefont {Barato}},\ and\ \bibinfo {author} {\bibfnamefont {U.}~\bibnamefont {Seifert}},\ }\bibfield  {title} {\bibinfo {title} {Sensory capacity: An information theoretical measure of the performance of a sensor},\ }\href@noop {} {\bibfield  {journal} {\bibinfo  {journal} {Physical Review E}\ }\textbf {\bibinfo {volume} {93}},\ \bibinfo {pages} {022116} (\bibinfo {year} {2016})}\BibitemShut {NoStop}%
\bibitem [{\citenamefont {Wark}\ \emph {et~al.}(2007)\citenamefont {Wark}, \citenamefont {Lundstrom},\ and\ \citenamefont {Fairhall}}]{wark2007sensory}%
  \BibitemOpen
  \bibfield  {author} {\bibinfo {author} {\bibfnamefont {B.}~\bibnamefont {Wark}}, \bibinfo {author} {\bibfnamefont {B.~N.}\ \bibnamefont {Lundstrom}},\ and\ \bibinfo {author} {\bibfnamefont {A.}~\bibnamefont {Fairhall}},\ }\bibfield  {title} {\bibinfo {title} {Sensory adaptation},\ }\href@noop {} {\bibfield  {journal} {\bibinfo  {journal} {Current opinion in neurobiology}\ }\textbf {\bibinfo {volume} {17}},\ \bibinfo {pages} {423} (\bibinfo {year} {2007})}\BibitemShut {NoStop}%
\bibitem [{\citenamefont {Lan}\ \emph {et~al.}(2012)\citenamefont {Lan}, \citenamefont {Sartori}, \citenamefont {Neumann}, \citenamefont {Sourjik},\ and\ \citenamefont {Tu}}]{lan2012energy}%
  \BibitemOpen
  \bibfield  {author} {\bibinfo {author} {\bibfnamefont {G.}~\bibnamefont {Lan}}, \bibinfo {author} {\bibfnamefont {P.}~\bibnamefont {Sartori}}, \bibinfo {author} {\bibfnamefont {S.}~\bibnamefont {Neumann}}, \bibinfo {author} {\bibfnamefont {V.}~\bibnamefont {Sourjik}},\ and\ \bibinfo {author} {\bibfnamefont {Y.}~\bibnamefont {Tu}},\ }\bibfield  {title} {\bibinfo {title} {The energy--speed--accuracy trade-off in sensory adaptation},\ }\href@noop {} {\bibfield  {journal} {\bibinfo  {journal} {Nature physics}\ }\textbf {\bibinfo {volume} {8}},\ \bibinfo {pages} {422} (\bibinfo {year} {2012})}\BibitemShut {NoStop}%
\bibitem [{\citenamefont {Conti}\ and\ \citenamefont {Mora}(2022)}]{conti2022nonequilibrium}%
  \BibitemOpen
  \bibfield  {author} {\bibinfo {author} {\bibfnamefont {D.}~\bibnamefont {Conti}}\ and\ \bibinfo {author} {\bibfnamefont {T.}~\bibnamefont {Mora}},\ }\bibfield  {title} {\bibinfo {title} {Nonequilibrium dynamics of adaptation in sensory systems},\ }\href@noop {} {\bibfield  {journal} {\bibinfo  {journal} {Physical Review E}\ }\textbf {\bibinfo {volume} {106}},\ \bibinfo {pages} {054404} (\bibinfo {year} {2022})}\BibitemShut {NoStop}%
\bibitem [{\citenamefont {Pittendrigh}(1954)}]{pittendrigh1954temperature}%
  \BibitemOpen
  \bibfield  {author} {\bibinfo {author} {\bibfnamefont {C.~S.}\ \bibnamefont {Pittendrigh}},\ }\bibfield  {title} {\bibinfo {title} {On temperature independence in the clock system controlling emergence time in drosophila},\ }\href@noop {} {\bibfield  {journal} {\bibinfo  {journal} {Proceedings of the National Academy of Sciences}\ }\textbf {\bibinfo {volume} {40}},\ \bibinfo {pages} {1018} (\bibinfo {year} {1954})}\BibitemShut {NoStop}%
\bibitem [{\citenamefont {Johnson}\ and\ \citenamefont {Rust}(2021)}]{johnson2021circadian}%
  \BibitemOpen
  \bibfield  {author} {\bibinfo {author} {\bibfnamefont {C.~H.}\ \bibnamefont {Johnson}}\ and\ \bibinfo {author} {\bibfnamefont {M.~J.}\ \bibnamefont {Rust}},\ }\href@noop {} {\emph {\bibinfo {title} {Circadian rhythms in bacteria and microbiomes}}},\ Vol.\ \bibinfo {volume} {409}\ (\bibinfo  {publisher} {Springer},\ \bibinfo {year} {2021})\BibitemShut {NoStop}%
\bibitem [{\citenamefont {Hogenesch}\ and\ \citenamefont {Ueda}(2011)}]{hogenesch2011understanding}%
  \BibitemOpen
  \bibfield  {author} {\bibinfo {author} {\bibfnamefont {J.~B.}\ \bibnamefont {Hogenesch}}\ and\ \bibinfo {author} {\bibfnamefont {H.~R.}\ \bibnamefont {Ueda}},\ }\bibfield  {title} {\bibinfo {title} {Understanding systems-level properties: timely stories from the study of clocks},\ }\href@noop {} {\bibfield  {journal} {\bibinfo  {journal} {Nature Reviews Genetics}\ }\textbf {\bibinfo {volume} {12}},\ \bibinfo {pages} {407} (\bibinfo {year} {2011})}\BibitemShut {NoStop}%
\bibitem [{\citenamefont {Ay}\ and\ \citenamefont {Krakauer}(2007)}]{ay2007geometric}%
  \BibitemOpen
  \bibfield  {author} {\bibinfo {author} {\bibfnamefont {N.}~\bibnamefont {Ay}}\ and\ \bibinfo {author} {\bibfnamefont {D.~C.}\ \bibnamefont {Krakauer}},\ }\bibfield  {title} {\bibinfo {title} {Geometric robustness theory and biological networks},\ }\href@noop {} {\bibfield  {journal} {\bibinfo  {journal} {Theory in biosciences}\ }\textbf {\bibinfo {volume} {125}},\ \bibinfo {pages} {93} (\bibinfo {year} {2007})}\BibitemShut {NoStop}%
\bibitem [{\citenamefont {Fu}\ \emph {et~al.}(2024)\citenamefont {Fu}, \citenamefont {Fei}, \citenamefont {Ouyang},\ and\ \citenamefont {Tu}}]{fu2024temperature}%
  \BibitemOpen
  \bibfield  {author} {\bibinfo {author} {\bibfnamefont {H.}~\bibnamefont {Fu}}, \bibinfo {author} {\bibfnamefont {C.}~\bibnamefont {Fei}}, \bibinfo {author} {\bibfnamefont {Q.}~\bibnamefont {Ouyang}},\ and\ \bibinfo {author} {\bibfnamefont {Y.}~\bibnamefont {Tu}},\ }\bibfield  {title} {\bibinfo {title} {Temperature compensation through kinetic regulation in biochemical oscillators},\ }\href@noop {} {\bibfield  {journal} {\bibinfo  {journal} {arXiv preprint arXiv:2401.13960}\ } (\bibinfo {year} {2024})}\BibitemShut {NoStop}%
\bibitem [{\citenamefont {Maes}(2020)}]{maes2020response}%
  \BibitemOpen
  \bibfield  {author} {\bibinfo {author} {\bibfnamefont {C.}~\bibnamefont {Maes}},\ }\bibfield  {title} {\bibinfo {title} {Response theory: a trajectory-based approach},\ }\href@noop {} {\bibfield  {journal} {\bibinfo  {journal} {Frontiers in Physics}\ }\textbf {\bibinfo {volume} {8}},\ \bibinfo {pages} {229} (\bibinfo {year} {2020})}\BibitemShut {NoStop}%
\bibitem [{\citenamefont {Maes}\ \emph {et~al.}(2013)\citenamefont {Maes}, \citenamefont {Safaverdi}, \citenamefont {Visco},\ and\ \citenamefont {Van~Wijland}}]{maes2013fluctuation}%
  \BibitemOpen
  \bibfield  {author} {\bibinfo {author} {\bibfnamefont {C.}~\bibnamefont {Maes}}, \bibinfo {author} {\bibfnamefont {S.}~\bibnamefont {Safaverdi}}, \bibinfo {author} {\bibfnamefont {P.}~\bibnamefont {Visco}},\ and\ \bibinfo {author} {\bibfnamefont {F.}~\bibnamefont {Van~Wijland}},\ }\bibfield  {title} {\bibinfo {title} {Fluctuation-response relations for nonequilibrium diffusions with memory},\ }\href@noop {} {\bibfield  {journal} {\bibinfo  {journal} {Physical Review E—Statistical, Nonlinear, and Soft Matter Physics}\ }\textbf {\bibinfo {volume} {87}},\ \bibinfo {pages} {022125} (\bibinfo {year} {2013})}\BibitemShut {NoStop}%
\bibitem [{\citenamefont {Baiesi}\ \emph {et~al.}(2009)\citenamefont {Baiesi}, \citenamefont {Maes},\ and\ \citenamefont {Wynants}}]{baiesi2009fluctuations}%
  \BibitemOpen
  \bibfield  {author} {\bibinfo {author} {\bibfnamefont {M.}~\bibnamefont {Baiesi}}, \bibinfo {author} {\bibfnamefont {C.}~\bibnamefont {Maes}},\ and\ \bibinfo {author} {\bibfnamefont {B.}~\bibnamefont {Wynants}},\ }\bibfield  {title} {\bibinfo {title} {Fluctuations and response of nonequilibrium states},\ }\href@noop {} {\bibfield  {journal} {\bibinfo  {journal} {Physical review letters}\ }\textbf {\bibinfo {volume} {103}},\ \bibinfo {pages} {010602} (\bibinfo {year} {2009})}\BibitemShut {NoStop}%
\bibitem [{\citenamefont {Baiesi}\ and\ \citenamefont {Maes}(2013)}]{baiesi2013update}%
  \BibitemOpen
  \bibfield  {author} {\bibinfo {author} {\bibfnamefont {M.}~\bibnamefont {Baiesi}}\ and\ \bibinfo {author} {\bibfnamefont {C.}~\bibnamefont {Maes}},\ }\bibfield  {title} {\bibinfo {title} {An update on the nonequilibrium linear response},\ }\href@noop {} {\bibfield  {journal} {\bibinfo  {journal} {New Journal of Physics}\ }\textbf {\bibinfo {volume} {15}},\ \bibinfo {pages} {013004} (\bibinfo {year} {2013})}\BibitemShut {NoStop}%
\bibitem [{\citenamefont {Seifert}\ and\ \citenamefont {Speck}(2010)}]{seifert2010fluctuation}%
  \BibitemOpen
  \bibfield  {author} {\bibinfo {author} {\bibfnamefont {U.}~\bibnamefont {Seifert}}\ and\ \bibinfo {author} {\bibfnamefont {T.}~\bibnamefont {Speck}},\ }\bibfield  {title} {\bibinfo {title} {Fluctuation-dissipation theorem in nonequilibrium steady states},\ }\href@noop {} {\bibfield  {journal} {\bibinfo  {journal} {EPL (Europhysics Letters)}\ }\textbf {\bibinfo {volume} {89}},\ \bibinfo {pages} {10007} (\bibinfo {year} {2010})}\BibitemShut {NoStop}%
\bibitem [{\citenamefont {Pagare}\ \emph {et~al.}(2024)\citenamefont {Pagare}, \citenamefont {Zhang}, \citenamefont {Zheng},\ and\ \citenamefont {Lu}}]{pagare2024stochastic}%
  \BibitemOpen
  \bibfield  {author} {\bibinfo {author} {\bibfnamefont {A.}~\bibnamefont {Pagare}}, \bibinfo {author} {\bibfnamefont {Z.}~\bibnamefont {Zhang}}, \bibinfo {author} {\bibfnamefont {J.}~\bibnamefont {Zheng}},\ and\ \bibinfo {author} {\bibfnamefont {Z.}~\bibnamefont {Lu}},\ }\bibfield  {title} {\bibinfo {title} {Stochastic distinguishability of markovian trajectories},\ }\href@noop {} {\bibfield  {journal} {\bibinfo  {journal} {The Journal of Chemical Physics}\ }\textbf {\bibinfo {volume} {160}} (\bibinfo {year} {2024})}\BibitemShut {NoStop}%
\bibitem [{\citenamefont {Zheng}\ and\ \citenamefont {Lu}(2025{\natexlab{a}})}]{zheng2025nonequilibrium}%
  \BibitemOpen
  \bibfield  {author} {\bibinfo {author} {\bibfnamefont {J.}~\bibnamefont {Zheng}}\ and\ \bibinfo {author} {\bibfnamefont {Z.}~\bibnamefont {Lu}},\ }\bibfield  {title} {\bibinfo {title} {Nonequilibrium macroscopic response relations for counting statistics},\ }\href@noop {} {\bibfield  {journal} {\bibinfo  {journal} {arXiv preprint arXiv:2511.02041}\ } (\bibinfo {year} {2025}{\natexlab{a}})}\BibitemShut {NoStop}%
\bibitem [{\citenamefont {Dechant}\ and\ \citenamefont {Sasa}(2020)}]{dechant2020fluctuation}%
  \BibitemOpen
  \bibfield  {author} {\bibinfo {author} {\bibfnamefont {A.}~\bibnamefont {Dechant}}\ and\ \bibinfo {author} {\bibfnamefont {S.-i.}\ \bibnamefont {Sasa}},\ }\bibfield  {title} {\bibinfo {title} {Fluctuation--response inequality out of equilibrium},\ }\href@noop {} {\bibfield  {journal} {\bibinfo  {journal} {Proceedings of the National Academy of Sciences}\ }\textbf {\bibinfo {volume} {117}},\ \bibinfo {pages} {6430} (\bibinfo {year} {2020})}\BibitemShut {NoStop}%
\bibitem [{\citenamefont {Zheng}\ and\ \citenamefont {Lu}(2025{\natexlab{b}})}]{zheng2025nonlinear}%
  \BibitemOpen
  \bibfield  {author} {\bibinfo {author} {\bibfnamefont {J.}~\bibnamefont {Zheng}}\ and\ \bibinfo {author} {\bibfnamefont {Z.}~\bibnamefont {Lu}},\ }\bibfield  {title} {\bibinfo {title} {Nonlinear response relations and fluctuation-response inequalities for nonequilibrium stochastic systems},\ }\href@noop {} {\bibfield  {journal} {\bibinfo  {journal} {arXiv preprint arXiv:2509.19606}\ } (\bibinfo {year} {2025}{\natexlab{b}})}\BibitemShut {NoStop}%
\bibitem [{\citenamefont {Zheng}\ and\ \citenamefont {Lu}(2025{\natexlab{c}})}]{zheng2025universal}%
  \BibitemOpen
  \bibfield  {author} {\bibinfo {author} {\bibfnamefont {J.}~\bibnamefont {Zheng}}\ and\ \bibinfo {author} {\bibfnamefont {Z.}~\bibnamefont {Lu}},\ }\bibfield  {title} {\bibinfo {title} {Universal response inequalities beyond steady states via trajectory information geometry},\ }\href@noop {} {\bibfield  {journal} {\bibinfo  {journal} {Physical Review E}\ }\textbf {\bibinfo {volume} {112}},\ \bibinfo {pages} {L012103} (\bibinfo {year} {2025}{\natexlab{c}})}\BibitemShut {NoStop}%
\bibitem [{\citenamefont {Zheng}\ and\ \citenamefont {Lu}(2025{\natexlab{d}})}]{zheng2025unified}%
  \BibitemOpen
  \bibfield  {author} {\bibinfo {author} {\bibfnamefont {J.}~\bibnamefont {Zheng}}\ and\ \bibinfo {author} {\bibfnamefont {Z.}~\bibnamefont {Lu}},\ }\bibfield  {title} {\bibinfo {title} {Unified linear fluctuation-response theory arbitrarily far from equilibrium},\ }\href@noop {} {\bibfield  {journal} {\bibinfo  {journal} {Physical Review E}\ }\textbf {\bibinfo {volume} {112}},\ \bibinfo {pages} {064103} (\bibinfo {year} {2025}{\natexlab{d}})}\BibitemShut {NoStop}%
\bibitem [{\citenamefont {Zheng}\ and\ \citenamefont {Lu}(2026{\natexlab{a}})}]{zheng2026thermodynamic}%
  \BibitemOpen
  \bibfield  {author} {\bibinfo {author} {\bibfnamefont {J.}~\bibnamefont {Zheng}}\ and\ \bibinfo {author} {\bibfnamefont {Z.}~\bibnamefont {Lu}},\ }\bibfield  {title} {\bibinfo {title} {Thermodynamic and kinetic bounds for finite-frequency fluctuation-response},\ }\href@noop {} {\bibfield  {journal} {\bibinfo  {journal} {arXiv preprint arXiv:2602.18631}\ } (\bibinfo {year} {2026}{\natexlab{a}})}\BibitemShut {NoStop}%
\bibitem [{\citenamefont {Kwon}\ \emph {et~al.}(2025)\citenamefont {Kwon}, \citenamefont {Chun}, \citenamefont {Park},\ and\ \citenamefont {Lee}}]{kwon2025fluctuation}%
  \BibitemOpen
  \bibfield  {author} {\bibinfo {author} {\bibfnamefont {E.}~\bibnamefont {Kwon}}, \bibinfo {author} {\bibfnamefont {H.-M.}\ \bibnamefont {Chun}}, \bibinfo {author} {\bibfnamefont {H.}~\bibnamefont {Park}},\ and\ \bibinfo {author} {\bibfnamefont {J.~S.}\ \bibnamefont {Lee}},\ }\bibfield  {title} {\bibinfo {title} {Fluctuation-response inequalities for kinetic and entropic perturbations},\ }\href@noop {} {\bibfield  {journal} {\bibinfo  {journal} {Physical Review Letters}\ }\textbf {\bibinfo {volume} {135}},\ \bibinfo {pages} {097101} (\bibinfo {year} {2025})}\BibitemShut {NoStop}%
\bibitem [{\citenamefont {Lee}\ \emph {et~al.}(2021)\citenamefont {Lee}, \citenamefont {Park},\ and\ \citenamefont {Park}}]{lee2021universal}%
  \BibitemOpen
  \bibfield  {author} {\bibinfo {author} {\bibfnamefont {J.~S.}\ \bibnamefont {Lee}}, \bibinfo {author} {\bibfnamefont {J.-M.}\ \bibnamefont {Park}},\ and\ \bibinfo {author} {\bibfnamefont {H.}~\bibnamefont {Park}},\ }\bibfield  {title} {\bibinfo {title} {Universal form of thermodynamic uncertainty relation for langevin dynamics},\ }\href@noop {} {\bibfield  {journal} {\bibinfo  {journal} {Physical Review E}\ }\textbf {\bibinfo {volume} {104}},\ \bibinfo {pages} {L052102} (\bibinfo {year} {2021})}\BibitemShut {NoStop}%
\bibitem [{\citenamefont {Dechant}(2025)}]{dechant2025finite}%
  \BibitemOpen
  \bibfield  {author} {\bibinfo {author} {\bibfnamefont {A.}~\bibnamefont {Dechant}},\ }\bibfield  {title} {\bibinfo {title} {Finite-frequency fluctuation-response inequality},\ }\href@noop {} {\bibfield  {journal} {\bibinfo  {journal} {arXiv preprint arXiv:2510.15228}\ } (\bibinfo {year} {2025})}\BibitemShut {NoStop}%
\bibitem [{\citenamefont {Hasegawa}\ and\ \citenamefont {Van~Vu}(2019)}]{hasegawa2019uncertainty}%
  \BibitemOpen
  \bibfield  {author} {\bibinfo {author} {\bibfnamefont {Y.}~\bibnamefont {Hasegawa}}\ and\ \bibinfo {author} {\bibfnamefont {T.}~\bibnamefont {Van~Vu}},\ }\bibfield  {title} {\bibinfo {title} {Uncertainty relations in stochastic processes: An information inequality approach},\ }\href@noop {} {\bibfield  {journal} {\bibinfo  {journal} {Physical Review E}\ }\textbf {\bibinfo {volume} {99}},\ \bibinfo {pages} {062126} (\bibinfo {year} {2019})}\BibitemShut {NoStop}%
\bibitem [{\citenamefont {Van~Vu}(2025)}]{van2025fundamental}%
  \BibitemOpen
  \bibfield  {author} {\bibinfo {author} {\bibfnamefont {T.}~\bibnamefont {Van~Vu}},\ }\bibfield  {title} {\bibinfo {title} {Fundamental bounds on precision and response for quantum trajectory observables},\ }\href@noop {} {\bibfield  {journal} {\bibinfo  {journal} {PRX Quantum}\ }\textbf {\bibinfo {volume} {6}},\ \bibinfo {pages} {010343} (\bibinfo {year} {2025})}\BibitemShut {NoStop}%
\bibitem [{\citenamefont {Aslyamov}\ \emph {et~al.}(2025)\citenamefont {Aslyamov}, \citenamefont {Ptaszy{\'n}ski},\ and\ \citenamefont {Esposito}}]{aslyamov2025nonequilibrium}%
  \BibitemOpen
  \bibfield  {author} {\bibinfo {author} {\bibfnamefont {T.}~\bibnamefont {Aslyamov}}, \bibinfo {author} {\bibfnamefont {K.}~\bibnamefont {Ptaszy{\'n}ski}},\ and\ \bibinfo {author} {\bibfnamefont {M.}~\bibnamefont {Esposito}},\ }\bibfield  {title} {\bibinfo {title} {Nonequilibrium fluctuation-response relations: From identities to bounds},\ }\href@noop {} {\bibfield  {journal} {\bibinfo  {journal} {Physical Review Letters}\ }\textbf {\bibinfo {volume} {134}},\ \bibinfo {pages} {157101} (\bibinfo {year} {2025})}\BibitemShut {NoStop}%
\bibitem [{\citenamefont {Aslyamov}\ \emph {et~al.}(2026)\citenamefont {Aslyamov}, \citenamefont {Ptaszy{\'n}ski},\ and\ \citenamefont {Esposito}}]{aslyamov2026macroscopic}%
  \BibitemOpen
  \bibfield  {author} {\bibinfo {author} {\bibfnamefont {T.}~\bibnamefont {Aslyamov}}, \bibinfo {author} {\bibfnamefont {K.}~\bibnamefont {Ptaszy{\'n}ski}},\ and\ \bibinfo {author} {\bibfnamefont {M.}~\bibnamefont {Esposito}},\ }\bibfield  {title} {\bibinfo {title} {Macroscopic fluctuation-response theory and its use for gene regulatory networks},\ }\href@noop {} {\bibfield  {journal} {\bibinfo  {journal} {Physical Review Letters}\ }\textbf {\bibinfo {volume} {136}},\ \bibinfo {pages} {067102} (\bibinfo {year} {2026})}\BibitemShut {NoStop}%
\bibitem [{\citenamefont {Ptaszy{\'n}ski}\ \emph {et~al.}(2026)\citenamefont {Ptaszy{\'n}ski}, \citenamefont {Aslyamov},\ and\ \citenamefont {Esposito}}]{ptaszynski2026nonequilibrium}%
  \BibitemOpen
  \bibfield  {author} {\bibinfo {author} {\bibfnamefont {K.}~\bibnamefont {Ptaszy{\'n}ski}}, \bibinfo {author} {\bibfnamefont {T.}~\bibnamefont {Aslyamov}},\ and\ \bibinfo {author} {\bibfnamefont {M.}~\bibnamefont {Esposito}},\ }\bibfield  {title} {\bibinfo {title} {Nonequilibrium fluctuation-response relations for state-current correlations},\ }\href@noop {} {\bibfield  {journal} {\bibinfo  {journal} {Physical Review E}\ }\textbf {\bibinfo {volume} {113}},\ \bibinfo {pages} {024131} (\bibinfo {year} {2026})}\BibitemShut {NoStop}%
\bibitem [{\citenamefont {Harunari}\ \emph {et~al.}(2024)\citenamefont {Harunari}, \citenamefont {Dal~Cengio}, \citenamefont {Lecomte},\ and\ \citenamefont {Polettini}}]{harunari2024mutual}%
  \BibitemOpen
  \bibfield  {author} {\bibinfo {author} {\bibfnamefont {P.~E.}\ \bibnamefont {Harunari}}, \bibinfo {author} {\bibfnamefont {S.}~\bibnamefont {Dal~Cengio}}, \bibinfo {author} {\bibfnamefont {V.}~\bibnamefont {Lecomte}},\ and\ \bibinfo {author} {\bibfnamefont {M.}~\bibnamefont {Polettini}},\ }\bibfield  {title} {\bibinfo {title} {Mutual linearity of nonequilibrium network currents},\ }\href@noop {} {\bibfield  {journal} {\bibinfo  {journal} {Physical Review Letters}\ }\textbf {\bibinfo {volume} {133}},\ \bibinfo {pages} {047401} (\bibinfo {year} {2024})}\BibitemShut {NoStop}%
\bibitem [{\citenamefont {Bebon}\ and\ \citenamefont {Speck}(2026)}]{bebon2026mutual}%
  \BibitemOpen
  \bibfield  {author} {\bibinfo {author} {\bibfnamefont {R.}~\bibnamefont {Bebon}}\ and\ \bibinfo {author} {\bibfnamefont {T.}~\bibnamefont {Speck}},\ }\bibfield  {title} {\bibinfo {title} {Mutual linearity is a generic property of steady-state markov networks},\ }\href@noop {} {\bibfield  {journal} {\bibinfo  {journal} {Physical Review Letters}\ }\textbf {\bibinfo {volume} {136}},\ \bibinfo {pages} {137401} (\bibinfo {year} {2026})}\BibitemShut {NoStop}%
\bibitem [{\citenamefont {Chun}\ and\ \citenamefont {Horowitz}(2023)}]{chun2023trade}%
  \BibitemOpen
  \bibfield  {author} {\bibinfo {author} {\bibfnamefont {H.-M.}\ \bibnamefont {Chun}}\ and\ \bibinfo {author} {\bibfnamefont {J.~M.}\ \bibnamefont {Horowitz}},\ }\bibfield  {title} {\bibinfo {title} {Trade-offs between number fluctuations and response in nonequilibrium chemical reaction networks},\ }\href@noop {} {\bibfield  {journal} {\bibinfo  {journal} {The Journal of Chemical Physics}\ }\textbf {\bibinfo {volume} {158}} (\bibinfo {year} {2023})}\BibitemShut {NoStop}%
\bibitem [{\citenamefont {Fernandes~Martins}\ and\ \citenamefont {Horowitz}(2023)}]{fernandes2023topologically}%
  \BibitemOpen
  \bibfield  {author} {\bibinfo {author} {\bibfnamefont {G.}~\bibnamefont {Fernandes~Martins}}\ and\ \bibinfo {author} {\bibfnamefont {J.~M.}\ \bibnamefont {Horowitz}},\ }\bibfield  {title} {\bibinfo {title} {Topologically constrained fluctuations and thermodynamics regulate nonequilibrium response},\ }\href@noop {} {\bibfield  {journal} {\bibinfo  {journal} {Physical Review E}\ }\textbf {\bibinfo {volume} {108}},\ \bibinfo {pages} {044113} (\bibinfo {year} {2023})}\BibitemShut {NoStop}%
\bibitem [{\citenamefont {Owen}\ \emph {et~al.}(2020)\citenamefont {Owen}, \citenamefont {Gingrich},\ and\ \citenamefont {Horowitz}}]{owen2020universal}%
  \BibitemOpen
  \bibfield  {author} {\bibinfo {author} {\bibfnamefont {J.~A.}\ \bibnamefont {Owen}}, \bibinfo {author} {\bibfnamefont {T.~R.}\ \bibnamefont {Gingrich}},\ and\ \bibinfo {author} {\bibfnamefont {J.~M.}\ \bibnamefont {Horowitz}},\ }\bibfield  {title} {\bibinfo {title} {Universal thermodynamic bounds on nonequilibrium response with biochemical applications},\ }\href@noop {} {\bibfield  {journal} {\bibinfo  {journal} {Physical Review X}\ }\textbf {\bibinfo {volume} {10}},\ \bibinfo {pages} {011066} (\bibinfo {year} {2020})}\BibitemShut {NoStop}%
\bibitem [{\citenamefont {Floyd}\ \emph {et~al.}(2025)\citenamefont {Floyd}, \citenamefont {Dinner},\ and\ \citenamefont {Vaikuntanathan}}]{floyd2025local}%
  \BibitemOpen
  \bibfield  {author} {\bibinfo {author} {\bibfnamefont {C.}~\bibnamefont {Floyd}}, \bibinfo {author} {\bibfnamefont {A.~R.}\ \bibnamefont {Dinner}},\ and\ \bibinfo {author} {\bibfnamefont {S.}~\bibnamefont {Vaikuntanathan}},\ }\bibfield  {title} {\bibinfo {title} {Local imperfect feedback control in non-equilibrium biophysical systems enabled by thermodynamic constraints},\ }\href@noop {} {\bibfield  {journal} {\bibinfo  {journal} {arXiv preprint arXiv:2507.07295}\ } (\bibinfo {year} {2025})}\BibitemShut {NoStop}%
\bibitem [{\citenamefont {Dal~Cengio}\ \emph {et~al.}(2025)\citenamefont {Dal~Cengio}, \citenamefont {Harunari}, \citenamefont {Lecomte},\ and\ \citenamefont {Polettini}}]{dal2025mutual}%
  \BibitemOpen
  \bibfield  {author} {\bibinfo {author} {\bibfnamefont {S.}~\bibnamefont {Dal~Cengio}}, \bibinfo {author} {\bibfnamefont {P.~E.}\ \bibnamefont {Harunari}}, \bibinfo {author} {\bibfnamefont {V.}~\bibnamefont {Lecomte}},\ and\ \bibinfo {author} {\bibfnamefont {M.}~\bibnamefont {Polettini}},\ }\bibfield  {title} {\bibinfo {title} {Mutual multilinearity of nonequilibrium network currents},\ }\href@noop {} {\bibfield  {journal} {\bibinfo  {journal} {SciPost Physics}\ }\textbf {\bibinfo {volume} {19}},\ \bibinfo {pages} {111} (\bibinfo {year} {2025})}\BibitemShut {NoStop}%
\bibitem [{\citenamefont {Khodabandehlou}\ \emph {et~al.}(2025)\citenamefont {Khodabandehlou}, \citenamefont {Maes},\ and\ \citenamefont {Neto{\v{c}}n{\`y}}}]{khodabandehlou2025affine}%
  \BibitemOpen
  \bibfield  {author} {\bibinfo {author} {\bibfnamefont {F.}~\bibnamefont {Khodabandehlou}}, \bibinfo {author} {\bibfnamefont {C.}~\bibnamefont {Maes}},\ and\ \bibinfo {author} {\bibfnamefont {K.}~\bibnamefont {Neto{\v{c}}n{\`y}}},\ }\bibfield  {title} {\bibinfo {title} {Affine relationships between steady currents},\ }\href@noop {} {\bibfield  {journal} {\bibinfo  {journal} {Journal of Physics A: Mathematical and Theoretical}\ }\textbf {\bibinfo {volume} {58}},\ \bibinfo {pages} {155002} (\bibinfo {year} {2025})}\BibitemShut {NoStop}%
\bibitem [{\citenamefont {Aslyamov}\ and\ \citenamefont {Esposito}(2024)}]{aslyamov2024general}%
  \BibitemOpen
  \bibfield  {author} {\bibinfo {author} {\bibfnamefont {T.}~\bibnamefont {Aslyamov}}\ and\ \bibinfo {author} {\bibfnamefont {M.}~\bibnamefont {Esposito}},\ }\bibfield  {title} {\bibinfo {title} {General theory of static response for markov jump processes},\ }\href@noop {} {\bibfield  {journal} {\bibinfo  {journal} {Physical Review Letters}\ }\textbf {\bibinfo {volume} {133}},\ \bibinfo {pages} {107103} (\bibinfo {year} {2024})}\BibitemShut {NoStop}%
\bibitem [{\citenamefont {Polettini}\ \emph {et~al.}(2025)\citenamefont {Polettini}, \citenamefont {Harunari}, \citenamefont {Cengio},\ and\ \citenamefont {Lecomte}}]{polettini2025coplanarity}%
  \BibitemOpen
  \bibfield  {author} {\bibinfo {author} {\bibfnamefont {M.}~\bibnamefont {Polettini}}, \bibinfo {author} {\bibfnamefont {P.~E.}\ \bibnamefont {Harunari}}, \bibinfo {author} {\bibfnamefont {S.~D.}\ \bibnamefont {Cengio}},\ and\ \bibinfo {author} {\bibfnamefont {V.}~\bibnamefont {Lecomte}},\ }\bibfield  {title} {\bibinfo {title} {Coplanarity of rooted spanning-tree vectors: M. polettini et al.},\ }\href@noop {} {\bibfield  {journal} {\bibinfo  {journal} {Letters in Mathematical Physics}\ }\textbf {\bibinfo {volume} {116}},\ \bibinfo {pages} {1} (\bibinfo {year} {2025})}\BibitemShut {NoStop}%
\bibitem [{\citenamefont {Meyer}(1962)}]{meyer1962decomposition}%
  \BibitemOpen
  \bibfield  {author} {\bibinfo {author} {\bibfnamefont {P.-A.}\ \bibnamefont {Meyer}},\ }\bibfield  {title} {\bibinfo {title} {A decomposition theorem for supermartingales},\ }\href@noop {} {\bibfield  {journal} {\bibinfo  {journal} {Illinois Journal of Mathematics}\ }\textbf {\bibinfo {volume} {6}},\ \bibinfo {pages} {193} (\bibinfo {year} {1962})}\BibitemShut {NoStop}%
\bibitem [{\citenamefont {Stutzer}\ \emph {et~al.}(2025)\citenamefont {Stutzer}, \citenamefont {Dieball},\ and\ \citenamefont {Godec}}]{stutzer2025stochastic}%
  \BibitemOpen
  \bibfield  {author} {\bibinfo {author} {\bibfnamefont {L.~T.}\ \bibnamefont {Stutzer}}, \bibinfo {author} {\bibfnamefont {C.}~\bibnamefont {Dieball}},\ and\ \bibinfo {author} {\bibfnamefont {A.}~\bibnamefont {Godec}},\ }\bibfield  {title} {\bibinfo {title} {Stochastic calculus for pathwise observables of markov-jump processes: Unification of diffusion and jump dynamics},\ }\href@noop {} {\bibfield  {journal} {\bibinfo  {journal} {arXiv preprint arXiv:2508.04647}\ } (\bibinfo {year} {2025})}\BibitemShut {NoStop}%
\bibitem [{\citenamefont {Peliti}\ and\ \citenamefont {Pigolotti}(2021)}]{peliti2021stochastic}%
  \BibitemOpen
  \bibfield  {author} {\bibinfo {author} {\bibfnamefont {L.}~\bibnamefont {Peliti}}\ and\ \bibinfo {author} {\bibfnamefont {S.}~\bibnamefont {Pigolotti}},\ }\href@noop {} {\emph {\bibinfo {title} {Stochastic thermodynamics: an introduction}}}\ (\bibinfo  {publisher} {Princeton University Press},\ \bibinfo {year} {2021})\BibitemShut {NoStop}%
\bibitem [{\citenamefont {Gillespie}(1977)}]{gillespie1977exact}%
  \BibitemOpen
  \bibfield  {author} {\bibinfo {author} {\bibfnamefont {D.~T.}\ \bibnamefont {Gillespie}},\ }\bibfield  {title} {\bibinfo {title} {Exact stochastic simulation of coupled chemical reactions},\ }\href@noop {} {\bibfield  {journal} {\bibinfo  {journal} {The journal of physical chemistry}\ }\textbf {\bibinfo {volume} {81}},\ \bibinfo {pages} {2340} (\bibinfo {year} {1977})}\BibitemShut {NoStop}%
\bibitem [{\citenamefont {Zheng}\ and\ \citenamefont {Lu}(2026{\natexlab{b}})}]{zheng2026mutual}%
  \BibitemOpen
  \bibfield  {author} {\bibinfo {author} {\bibfnamefont {J.}~\bibnamefont {Zheng}}\ and\ \bibinfo {author} {\bibfnamefont {Z.}~\bibnamefont {Lu}},\ }\bibfield  {title} {\bibinfo {title} {Mutual linearity in nonequilibrium langevin dynamics},\ }\href@noop {} {\bibfield  {journal} {\bibinfo  {journal} {arXiv preprint arXiv:2605.07949}\ } (\bibinfo {year} {2026}{\natexlab{b}})}\BibitemShut {NoStop}%
\bibitem [{\citenamefont {Dieball}\ and\ \citenamefont {Godec}(2023)}]{dieball2023direct}%
  \BibitemOpen
  \bibfield  {author} {\bibinfo {author} {\bibfnamefont {C.}~\bibnamefont {Dieball}}\ and\ \bibinfo {author} {\bibfnamefont {A.}~\bibnamefont {Godec}},\ }\bibfield  {title} {\bibinfo {title} {Direct route to thermodynamic uncertainty relations and their saturation},\ }\href@noop {} {\bibfield  {journal} {\bibinfo  {journal} {Physical Review Letters}\ }\textbf {\bibinfo {volume} {130}},\ \bibinfo {pages} {087101} (\bibinfo {year} {2023})}\BibitemShut {NoStop}%
\bibitem [{\citenamefont {Kwon}\ and\ \citenamefont {Lee}(2025)}]{kwon2025unified}%
  \BibitemOpen
  \bibfield  {author} {\bibinfo {author} {\bibfnamefont {E.}~\bibnamefont {Kwon}}\ and\ \bibinfo {author} {\bibfnamefont {J.~S.}\ \bibnamefont {Lee}},\ }\bibfield  {title} {\bibinfo {title} {A unified framework for classical and quantum uncertainty relations using stochastic representations},\ }\href@noop {} {\bibfield  {journal} {\bibinfo  {journal} {Communications Physics}\ }\textbf {\bibinfo {volume} {8}},\ \bibinfo {pages} {444} (\bibinfo {year} {2025})}\BibitemShut {NoStop}%
\bibitem [{\citenamefont {Zheng}(2026)}]{data}%
  \BibitemOpen
  \bibfield  {author} {\bibinfo {author} {\bibfnamefont {J.}~\bibnamefont {Zheng}},\ }\href@noop {} {} (\bibinfo {year} {2026}),\ \bibinfo {note} {https://github.com/Axeho2/Mutual-Linearity-in-and-out-of-Stationarity-for-Markov-Jump-Processes-A-Trajectory-Based-Approach.}\BibitemShut {Stop}%
\end{thebibliography}%

\end{document}